\documentclass[lettersize,journal]{IEEEtran}
\usepackage{amsmath,amsfonts}
\usepackage{algorithm}
\usepackage{algpseudocode}
\usepackage{amsthm}
\usepackage{array}
\usepackage[caption=false,font=normalsize,labelfont=sf,textfont=sf]{subfig}
\usepackage{textcomp}
\usepackage{stfloats}
\usepackage{url}
\usepackage{verbatim}
\usepackage{graphicx}
\usepackage{cite}
\usepackage{hyperref}
\usepackage{dsfont}

\usepackage{color}
\definecolor{turquoise}{cmyk}{0.65,0,0.1,0.3}
\definecolor{purple}{rgb}{0.65,0,0.65}
\definecolor{dark_green}{rgb}{0, 0.5, 0}
\definecolor{green}{rgb}{0, 1.0, 0}
\definecolor{orange}{rgb}{0.8, 0.6, 0.2}
\definecolor{red}{rgb}{0.8, 0.2, 0.2}
\definecolor{blueish}{rgb}{0.0, 0.7, 1}
\definecolor{light_gray}{rgb}{0.7, 0.7, .7}
\definecolor{pink}{rgb}{1, 0, 1}

\newcommand{\hide}[1]{{}} 

\newcommand{\etal}{{\em et al.}}
\newcommand{\eg}{{\em e.g.}}
\newcommand{\ie}{{\em i.e.}}

\newcommand{\Log}[1]{\log\left(#1\right)}

\newcommand{\Norm}[1]{\left\lVert#1\right\rVert}
\newcommand{\Abs}[1]{\left|#1\right|}

\newcommand{\Indicator}[1]{\mathds{1} \left\{ #1 \right\} }

\newtheorem{proposition}{Proposition}
\newtheorem*{remark}{Remark}
\newcommand{\LPIPS}{\textrm{LPIPS}}
\newcommand{\sLPIPS}{\textrm{sLPIPS}}
\newcommand{\beginsupplement}{%
        \setcounter{table}{0}
        \renewcommand{\thetable}{S\arabic{table}}%
        \setcounter{figure}{0}
        \renewcommand{\thefigure}{S\arabic{figure}}%
     }

\begin{document}

\title{Sandwiched Compression: Repurposing \protect\\ Standard Codecs with Neural Network Wrappers}

\author{Onur G.\ Guleryuz, Philip A.\ Chou, Berivan Isik, Hugues Hoppe, Danhang Tang, \\ Ruofei Du, Jonathan Taylor, Philip Davidson, and Sean Fanello
\thanks{This work was done while the authors were at Google Research. The source code for this work can be found at \cite{sandwich_oss}.}%
}

\markboth{IEEE Transaction on XXXX, 2024}%
{Shell \MakeLowercase{\textit{\etal}}: A Sample Article Using IEEEtran.cls for IEEE Journals}

\maketitle

\begin{abstract}

We propose sandwiching standard image and video codecs between pre- and post-processing neural networks.
The networks are jointly trained through a differentiable codec proxy to minimize a given rate-distortion loss.
This sandwich architecture not only improves the standard codec's performance on its intended content,
but more importantly, adapts
the codec to other types of image/video content and to other distortion measures.
The sandwich learns to transmit ``neural code images''
that optimize and improve overall rate-distortion performance, with the improvements becoming significant especially when the overall problem is well outside of the scope of the codec's design.
We apply the sandwich architecture to standard codecs with mismatched
sources transporting
different numbers of channels, higher resolution, higher dynamic range, computer graphics,
and with perceptual distortion measures.
The results demonstrate substantial improvements (up to 9 dB gains or up to 30\% bitrate reductions)
compared to alternative adaptations. 
We establish optimality properties for sandwiched compression and design differentiable codec proxies approximating current standard codecs.
We further analyze model complexity, visual quality under perceptual metrics,
as well as sandwich configurations that offer interesting potentials in video compression and streaming.
\end{abstract}

\begin{IEEEkeywords}
Image/video/computer-graphics
compression,
differentiable proxy,
rate-distortion optimization,
multi-spectral,
super-resolution,
high dynamic range,
perceptual distortion 
\vspace{-3mm}
\end{IEEEkeywords}


\section{Introduction}

\begin{figure}
\footnotesize
\begin{minipage}[b]{0.485\linewidth}
  \centering
  \includegraphics[width=4cm, trim=0 0 0 10mm, clip]{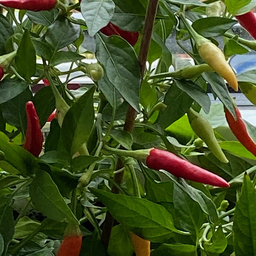}\\
  (a) Original source image
\end{minipage}
\begin{minipage}[b]{0.485\linewidth}
  \centering
  \includegraphics[width=4cm, trim=0 0 0 10mm, clip]{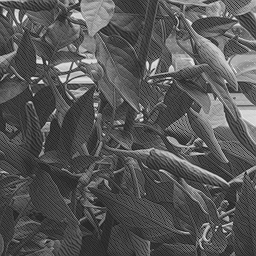}\\
  (b) Bottleneck (neural code) image
\end{minipage}

\begin{minipage}[b]{0.485\linewidth}
  \centering
  \includegraphics[width=4cm, trim=0 0 0 10mm, clip]{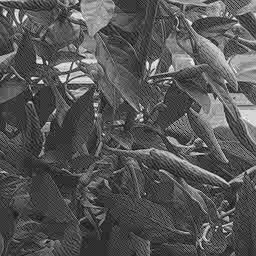}\\
  (c) Reconstructed bottleneck image
\end{minipage}
\begin{minipage}[b]{0.485\linewidth}
  \centering
  \includegraphics[width=4cm, trim=0 0 0 10mm, clip]{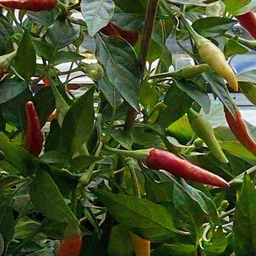}\\
  (d) Reconstructed source image
\end{minipage}
\vspace{-2mm}
\caption{
The sandwich architecture can accomplish %
surprising
results even with a simple codec (here JPEG 4:0:0, a single-channel grayscale codec).
The neural pre-processor is able to encode the full RGB image in (a) into a 
grayscale
image of neural codes in (b).
The neural codes are low-frequency dither-like patterns that modulate the color information yet also survive JPEG compression (c).
At the decoding end, the neural post-processor
demodulates the patterns to faithfully recover the 
color while also achieving deblocking.
The interested reader can generate an extensive set of further examples using our software at \cite{sandwich_oss}.}
\label{fig:modulation_dithering1}
\vspace{-7mm}
\end{figure}

Image and video compression are well-established domains, with a rich history marked by the evolution of 
standard codecs, such as JPEG, 
MPEG
1,2,4, 
H264/AVC, VC1,  VP9, H265/HEVC, AV1, and H266/VVC \cite{jpeg,wiegand2003overview,vp9,hevc,han2021technical,bross2021vvc}. These 
codecs are fundamentally rooted in linear transforms like the discrete cosine tranform (DCT) in the spatial dimensions, and 
motion-compensated prediction in the temporal dimension.  
Their designs and optimizations are 
typically guided by the analytically convenient mean-squared-error (MSE) metric albeit with eventual subjective quality verification.

Recent advancements have spotlighted {\em learned} image and video codecs based on neural networks trained end-to-end that are competitive with or outperform the standard codecs in rate-distortion metrics when distortion is measured through MSE \cite{MiBaTo18,BalleEtAl:20,GuoZFC:21,
lu2019dvc,hu2021fvc,rippel2021elf,MinnenJ:23}.
In other scenarios where complete faith to the source is not demanded,
more significant reductions in bitrate 
for the same visual quality have been obtained
by training networks using auxiliary distortion measures
(such as image likelihood modeled by discriminators in generative adversarial networks \cite{mentzer2020high}).
Neural networks have an obvious functional advantage over standard codecs in that they can be trained on datasets of images whose distributions are mismatched from the usual photographic images, for example medical, multispectral, depth, geometric, or other unusual image classes, as well as other distortion criteria, including human perceptual criteria but also machine performance criteria (\eg, classification, segmentation, labeling, and diagnosis.)

Unfortunately, the performance and functional advantages of neural codecs come at great computational cost
\cite{BalleEtAl:20,lu2019dvc,hu2021fvc,rippel2021elf,MinnenJ:23}.
This cost is typically at
a level that
is impractical
for HD imagery at video rates even on dedicated neural chips especially in mobile devices where power is a prime concern.
Handling UHD at graphics rates is even more impractical. Capable networks require massive computational resources, power, and chip area. {\em In-loop} tandem GPU-CPU solutions where part of the CPU (GPU) compute is offloaded to the GPU (CPU) need massive bandwidth. As of this writing, even for next generation compression standards under development \cite{AV2}, these issues force neural tools to be limited to 1000 Multiply-Accumulate (MAC) networks. Competent neural models need {\em hundreds of thousands to millions} of MACs \cite{nick_icip, neural_recent_cvpr}. A 1000 MAC model is understandably quite limited in comparison \cite{Kiran_Andrew}.
\begin{figure}[t]
    \centering
\begin{minipage}{0.27\linewidth}
  \centering
  \includegraphics[width=\linewidth]{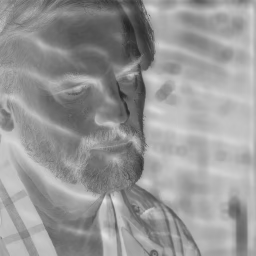}
\end{minipage} 
\begin{minipage}{0.27\linewidth}
  \centering
  \includegraphics[width=\linewidth]{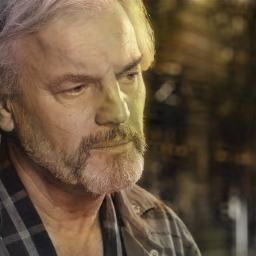}
\end{minipage} 
\begin{minipage}{0.27\linewidth}
  \centering
  \includegraphics[width=\linewidth]{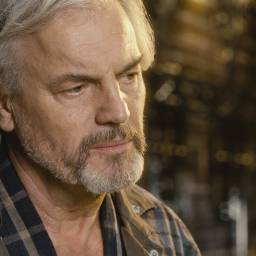} 
\end{minipage} 
\\
\vspace{1mm}
\begin{minipage}{0.27\linewidth}
  \centering
  \includegraphics[width=\linewidth]{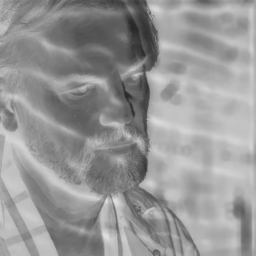}
\end{minipage} 
\begin{minipage}{0.27\linewidth}
  \centering
  \includegraphics[width=\linewidth]{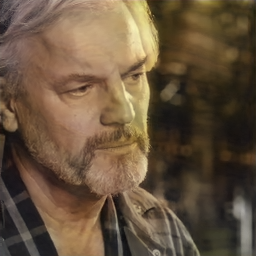}
\end{minipage} 
\begin{minipage}{0.27\linewidth}
  \centering
  \includegraphics[width=\linewidth]{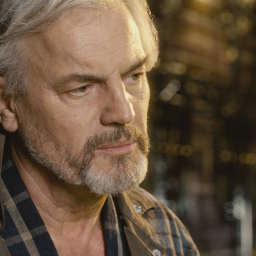}
\end{minipage} 
\\
\vspace{1mm}
\begin{minipage}{0.27\linewidth}
  \centering
  \includegraphics[width=\linewidth]{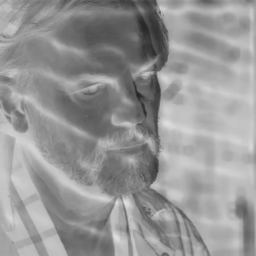}
\end{minipage} 
\begin{minipage}{0.27\linewidth}
  \centering
  \includegraphics[width=\linewidth]{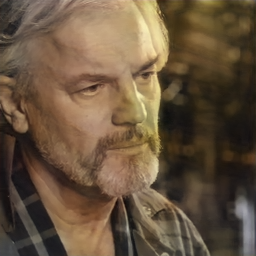}
\end{minipage} 
\begin{minipage}{0.27\linewidth}
  \centering
  \includegraphics[width=\linewidth]{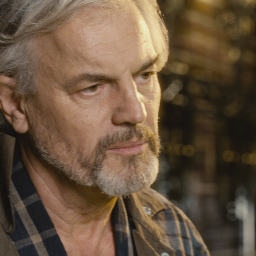}
\end{minipage}
    \caption{Analogue of \autoref{fig:modulation_dithering1} for video and HEVC.
The sandwich is used to transport full color video over a gray-scale codec (HEVC 4:0:0). First, fifth, and tenth frames of compressed bottlenecks, final reconstructions by the
post-processor, and original source videos are shown. Rate=0.07 bpp, PSNR=36.0 dB. The sandwich establishes temporally coherent modulation-like patterns on the bottlenecks through which the pre-processor encodes color
that are then demodulated by the post-processor for a full-color result. The patterns are spatially broader compared to those in \autoref{fig:modulation_dithering1} to facilitate more efficient motion compensation.
The interested reader can generate an extensive set of further examples using our software at \cite{sandwich_oss}.}
\vspace{-5mm}
\label{fig:hevc_400}
\end{figure}

\newcommand{\propositionautorefname}{Proposition}
An interesting way of using neural networks is in the form of pre-post processors that function outside of the main compression loop \cite{QiuLD21}. Such designs do not suffer from bandwidth issues as CPU-GPU transfers are one-way and can be made optional, for example, only enabled for capable decoders with neural processing resources. As we show in this paper 
(see results with perceptual distortion measures later) one can even augment the standards-based compression chain with a single neural pre-processor and accomplish significant performance bumps while targeting standard decoders unaware/incapable of neural processing\footnote{Of course, as discussed in \autoref{sec:prelude}, even end-to-end, purpose-built, minimal-weight neural codecs can be augmented with a sandwich to re-purpose them to different content or distortion measures.}.
Another advantage of the pre-post processors as proposed in this paper 
is their highly parallelizable nature, with parallelism easily exploited by, say, spatially tiling the picture over GPU slices to take already above-real-time performance to many hundreds if not thousands of frames-per-second \cite{nvidia_slice}.
In comparison, neural processing is typically limited to a few frames-per-second \cite{neural_recent_cvpr}.

In this paper, we propose such a pre-post architecture as a \emph{hybrid} between standard codecs and purely neural codecs, which we call the \emph{sandwich architecture}\footnote{Our early work appeared in \cite{GuleryuzCHTDDF:21,guleryuz2022sandwiched,IsikGTTC:23}.}.  In the sandwich architecture, a standard codec is positioned between a neural pre-processor and a neural post-processor, which are {\em jointly} trained to minimize distortion subject to a bitrate constraint.  The neural pre- and post-processors can be lightweight, yet they are able to improve the rate-distortion performance of the standard codec, even on typical color photographic images/video when the distortion measure is MSE.  
Much more interestingly, we show that the improvement is 
especially
pronounced when the application mismatches the standard codec's design target in some way, including non-RGB images (\eg, $C$-channel images with $C\neq3$ such as medical, multi-spectral, depth, geometric, and other sensed images), non-MSE distortion criteria (\eg, human perceptual metrics, realism metrics, and machine task-specific performance metrics), and non-standard profile hardware constraints (\eg, higher bit depth and higher spatial resolution).  At the same time, the sandwich approach leverages the standard codec for much of the heavy lifting, including highly efficient transforms, entropy coding, and motion processing.  Vast resources have been put into
the hardware implementations of standard codecs and
their broader ecosystems (transparent packetization, networking, routing, etc.),
which can make these resources essentially free compared to the power consumption required in neural chips.

As we illustrate in this paper, the magic behind the sandwich architecture is the ability for the pre-processor to learn how to produce images of {\em neural codes} that are well-compressible by a standard codec 
and for the post-processor to learn how to decode these images, to minimize the relevant distortion measure.
Of course,  the neural code images have to be robust to the compression noise 
typically inserted 
at those bitrates by the standard codec.
To gain an intuitive understanding of what these neural code images may look like, consider a simplified problem in which the neural pre- and post-processors adapt a 1-channel (grayscale) codec to compressing an ordinary 3-channel color image, shown in \autoref{fig:modulation_dithering1}.  The top image (a) is the {\em original source image} fed into the pre-processor.  The neural code image (b), or {\em bottleneck image}, is the image produced by the pre-processor and fed into the encoder of the standard codec.  This is called the bottleneck image since it is at the locus of the compression bottleneck.  Note that the bottleneck image contains spatial modulation patterns (akin to watermarks) that serve to encode the color information in this case.  These patterns are the neural codes.  
The {\em reconstructed bottleneck image} (c) is emitted from decoder of the standard codec (note the typical JPEG blocking artifacts) and fed into the post-processor.  The {\em reconstructed source image} (d) is the image output from the post-processor.  Note that both color and sharp spatial definition are recovered from the neural codes\footnote{The reader versed in  
watermarking and data-embedding  \cite{moulin_capacity} will note the similarities except that the processors in this case 
need not
hide the embedded data. That the networks have to be jointly optimized 
is clear.}.
\autoref{fig:hevc_400} shows the analogue for video compression with HEVC. This time the processors use temporally coherent modulation patterns to communicate color.
(Refer to
subsections \ref{sec:rgb-images} and \ref{sec:video-400}
for rate-distortion results relevant to these scenarios, and to \ref{sec:video-400} for a discussion on neural codes' temporal coherence.) 

Few works prior to 
ours
paired standard codecs with neural networks as pre- or post-processors.  Most 
paired the standard codec with either a neural pre-processor alone (\eg, to perform denoising of the input image \cite{ZhangZCMZ:17,TIAN2020251,vu2020unrolling}) or a neural post-processor alone (\eg, to perform deblocking or other enhancements of the output image \cite{Svoboda2016CompressionAR,KimLSL:19,Niu2020EndtoEndJD}).  A few works paired standard codecs with both neural pre- and post-processors, such as \cite{LiLLLLW19,EusebioAP20,QiuLD21,Andreopoulos22,Kim2020a}, but these solutions, like prior non-neural solutions such as \cite{SegallK00} 
did so in such a way that the pre- and post-processors may be used independently; thus no neural codes are generated; hence they do not take full advantage of the communication available between the pre- and post-processors (see \autoref{thm:proposition1} as to why this is critically important.)

Beyond proposing the sandwich architecture itself, a major contribution of our paper is a solution for \emph{jointly training the pre- and post-processors}.  To jointly train these neural networks using standard gradient back-propagation, the standard codec must be differentiable.  Hence during training we replace the standard codec with a differentiable {\em codec proxy}.  We show that well-designed simple proxies that approximate key codec components allow the training of models that robustly work with different standard codecs. 
Using these proxies we demonstrate the
\emph{advantages of the sandwich architecture across a variety of image and video compression settings}:

$\bullet$ For coding of 3-channel color images over a 1-channel (grayscale, or 4:0:0) codec, as in \autoref{fig:modulation_dithering1}, sandwiching has 6--9 dB gain in MSE.
Over a 1.5-channel (4:2:0) codec, sand\-wiching has a
10\% reduction in bitrate.
And over a 3-channel (4:4:4) codec, sandwiching has a
15\% reduction in bitrate.

$\bullet$ For coding of 2x high resolution (HR) or super-resolved images over 1x lower resolution (LR) codecs, sandwiching has up to 9 dB gain in MSE.

$\bullet$ For coding of 16-bit high dynamic range\footnote{In this work, the term HDR is interchangeable with high-bit-depth.} (HDR) color images over an 8-bit lower dynamic range (LDR) codec, sandwiching has up to 3 dB gain in MSE.

$\bullet$ For coding of 3-channel computer graphics normal maps over a 1.5-channel codec,
sandwiching has a 4-5 dB gain in MSE.
Over a 3-channel codec, sandwiching has a 1.5-2 dB gain, or about 15\% reduction in bitrate.

$\bullet$ For coding of 8-channel computer graphics texture maps (3-channel albedo, 3-channel normals, 1-channel roughness, and 1-channel occlusion) over an 8-channel codec (implemented as the concatenation of eight 1-channel instances), sandwiching has a 20-30\% reduction in bitrate, when distortion is measured %
over the final
lighted, rendered images across multiple views.
We term this measure the {\em shaded distortion}.

$\bullet$ For coding of video, we show analogous gains in MSE for coding color over grayscale codecs,
and coding HR over LR codecs.
Perhaps most importantly from the perspective of video
applications,
we demonstrate that for coding color video over color codecs, sandwiching yields 30\% reduction in bitrate at the same visual quality, when trained to minimize the perceptual distortion measure LPIPS 
instead of MSE. We also provide related VMAF results.

Our results are geared toward establishing the following outline.
In order to generate an intuitive understanding and to analyze the role of encoding with different sub-sampling patterns (4:4:4 v.s. 4:2:0 and so on) we start with image compression with a straightforward transform coder as embodied by the JPEG standard. 
We then demonstrate that the results do not depend heavily on whether the standard codec is JPEG or HEVC-Intra (HEIC) without any model retraining indicating the efficacy of our proxies.
Our primary results are with HEVC as the codec-du-jour that immediately benefits from many re-purposing scenarios we look at.
We also show that the results degrade but hold up well as the number of parameters of the
pre- and post-processors is reduced by more than two orders of magnitude, i.e., a 99\% reduction in parameters, 
in each neural processor. 
The main point we establish is the capacity of the sandwich in re-purposing the standard codec in various applications with very significant improvements.
We nevertheless further point to work that successfully uses the sandwich for basic compression improvements with VVC/AV1 on HD/UHD video using very low complexity models.

\renewcommand{\sectionautorefname}{Section}

The rest of the paper is organized as follows. 
The prelude of \autoref{sec:prelude} %
discusses the optimal sandwich, mathematically framing our work within the rich pre-post-processor literature. 
\autoref{sec:setup} presents the sandwich architecture.  \autoref{sec:images}
includes image compression experiments
followed by complexity results in
\autoref{sec:complexity}. 
\autoref{sec:video}
is devoted to 
video compression experiments.
\autoref{sec:discussion} concludes the paper.
\vspace{-3mm}

\section{Prelude: The Sandwich as a Codelength Constrained Vector Quantizer}
\label{sec:prelude}

Pre-Post processing applied around 
a compression codec is a well-known technique. In $\Sigma\Delta$ compression \cite{sigma_delta} one wraps a simple $1$-bit quantizer to 
make it function like a $k$-bit one, in \cite{guleryuz_dpcm} one wraps DPCM codecs (performance-wise inferior to transform codecs) to get them to perform like transform codecs, using \cite{ZhangZCMZ:17,TIAN2020251,vu2020unrolling, Svoboda2016CompressionAR,KimLSL:19,Niu2020EndtoEndJD} one can wrap image/video codecs to reduce input noise, reduce codec artifacts, and so on.  
Compression literature includes many such interesting designs that
offer specific solutions to specific problems. With neural networks one now has the capability of designing much more general mappings as %
pre-post processors. In this section we briefly explore the potential gains one can tap into.

Compression codecs can be seen as vector quantizer code-books. A standard codec at a particular operating point can 
be thought of in terms of a set of codewords (decoder reconstruction vectors existing in high dimensions) and associated binary strings (bits signaling each desired reconstruction). A sandwich with non-identity wrappers maps a source
to use the standard codebook and then maps the standard decoder's output into final reconstructions.  
Looking from outside the sandwich, we hence see a new codebook for the source that is determined by the pre/post-processor mappings modifying the standard codebook. Suppose the standard codec is not adequate for a given source. Then, a natural question to ask is ``{\em How much better can we make the standard codec by wrapping it in a sandwich?}''

In order to quantify the properties of ``sandwich-achievable'' codebooks and how they would compare to a codebook that is optimal for the source,
let us momentarily disregard limitations on neural network complexity and limitations of back propagation in finding overall optimal solutions. Assume we can find the optimal pre-post-processor mappings. What is the efficiency of the sandwich system with respect to an optimal codebook? The following proposition shows that {\em the optimal sandwich can accomplish the optimal compression performance except for a potential rate penalty induced by a  mismatch to the standard codec's codelengths}.

\begin{proposition}
\label{thm:proposition1}[Optimal Sandwich]
Let $X$ be a $\mathbb{R}^n$-valued bounded source, let $d$ be a distortion measure, and let $D(R)$ be the operational distortion-rate function for $X$ under $d$.  For any $\epsilon>0$, let $(\alpha^*,\beta^*,\gamma^*)$ be the encoding, decoding, and lossless coding maps for a rate-$R$ codec for $X$ achieving $D(R)$ within $\epsilon/2$.
Let
$(\alpha,\beta,\gamma)$ be a regular codec (\eg, a standard codec, possibly designed for a different source and different distortion measure) with bounded codelengths.  Then there exist neural pre- and post-processors $f$ and $g$ such that the codec sandwich $(\alpha\circ f,g\circ\beta,\gamma)$ has expected distortion at most $D(R)+\epsilon$ and expected rate at most $R+D(p||q)+\epsilon$, where $p(k)=P(\{\alpha^*(X)=k\})$ and $q(k)=2^{-|\gamma(k)|}$.
\end{proposition}
\begin{proof}
See \autoref{sec:theory} in Supplementary Material.
\vspace{-2mm}
\end{proof}

Note the key role of the sandwich in 
repurposing 
the inner codebook to the outer compression scenario. When sandwiching a high-performance image/video codec for different but related image/video applications one can expect the mismatch to be lighter compared to, say, when one tries to sandwich an image codec to transport audio data.
From the perspective of the proposition, using 
configurable codecs, \ie, those that allow codebook codelengths to be optimized,
may help minimize the implied penalty.
While beyond the scope of this paper, we 
point to
generalizing the sandwich to configurable codecs as an interesting research direction.
\vspace{-3mm}

\section{The Sandwich Architecture}
\label{sec:setup}

\subsection{Sandwich for Image Compression}
\label{sec:image_sandwich}

The sandwich architecture for image compression is shown in \autoref{fig:architecture}(a).  An {\em original source image} $S$ with one or more full-resolution channels is mapped by a neural pre-processor into one or more channels of {\em neural} (or {\em latent}) {\em codes}.  Each channel of neural codes may be 
full resolution or reduced resolution.  The channels of neural codes are grouped into one or more {\em bottleneck images} $B$ suitable for consumption by a standard image codec.  The bottleneck images are compressed by the standard image encoder into a bitstream, which
is decompressed by the corresponding decoder into {\em reconstructed bottleneck images} $\hat B$.  The channels of the reconstructed bottleneck images are then mapped by a neural post-processor into a {\em reconstructed source image} $\hat S$.

\begin{figure}
\footnotesize
\begin{minipage}[b]{0.9\linewidth}
  \centering
  \includegraphics[width=0.99\linewidth, trim=0 150 290 0cm, clip]{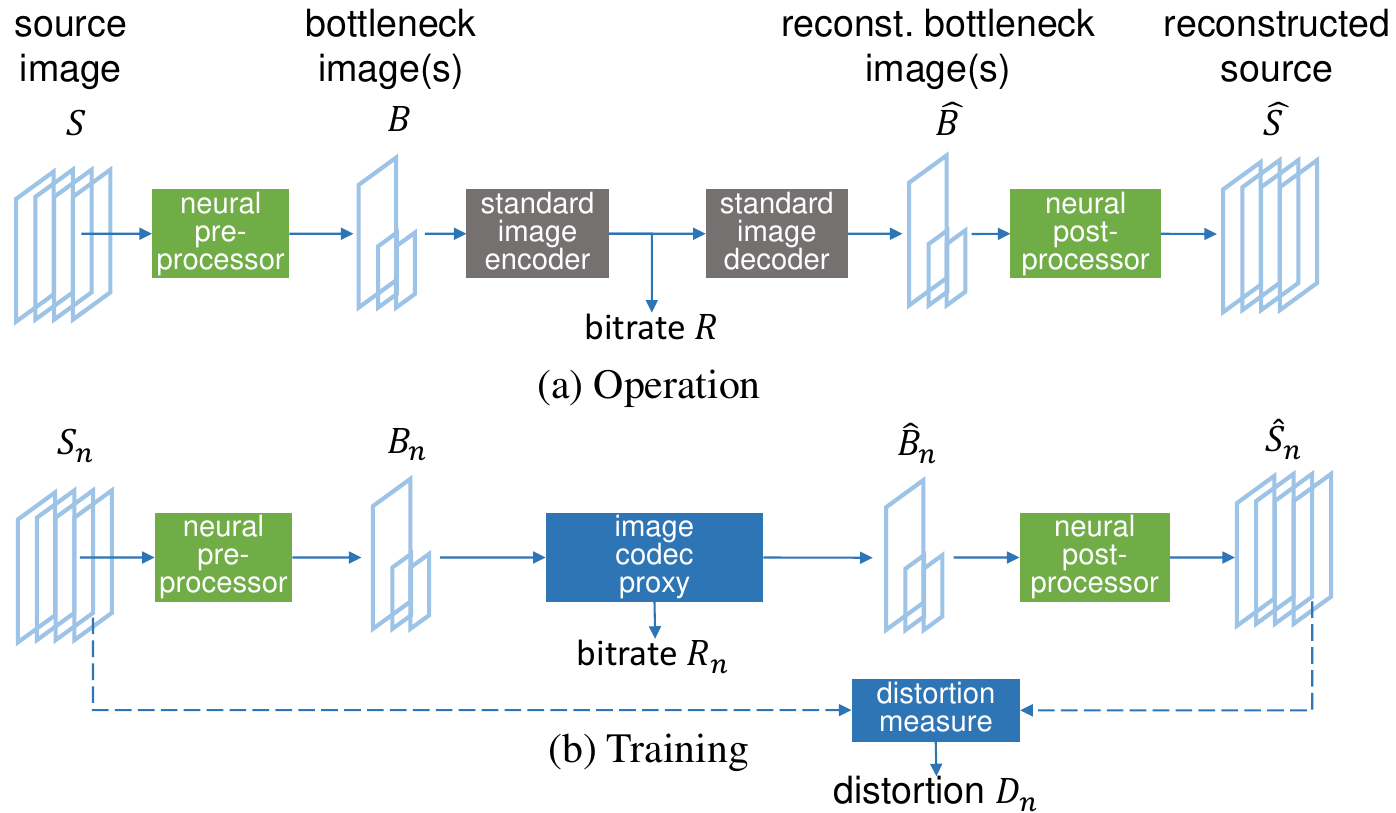}
\end{minipage}
\vspace{-6mm}
\caption{Neural-sandwiched image codec during (a) operation and (b) training. Gray boxes are not differentiable; blue are differentiable; green are trainable.  Loss function for training is $\sum D_n+\lambda R_n$ over example images $n$.}
\vspace{-4mm}
\label{fig:architecture}
\end{figure}

The standard image codec in the sandwich is configured to \emph{avoid} any color conversion or further subsampling.
Thus, it compresses three full-resolution channels as an image in 4:4:4 format, one full-resolution channel and two half-resolution channels as an image in 4:2:0 format, or one full-resolution channel as an image in 4:0:0 (\ie, grayscale) format --- all without color conversion.
Other combinations of channels are processed by appropriate grouping.

\autoref{fig:processor} shows the network architectures we use for our neural pre-processor and post-processor.
The upper branch of the network learns pointwise operations, like color conversion, using a multilayer perceptron (MLP) or equivalently a series of $1\!\times\!1$ 2D convolutional layers, while the lower branch uses a U-Net \cite{UNet:15} to take into account more complex spatial context.
At the output of the pre-processor, any half-resolution channels are obtained by sub-sampling, while at the input of the post-processor, any half-resolution channels are first upsampled to full resolution. %
We have deliberately picked the U-Net as it is a well-known model whose performance in various areas is well-documented. U-Nets %
have also been systematically studied with reduced parameter/complexity variants easily generated.

\begin{figure}%
\footnotesize
\begin{minipage}[b]{0.48\linewidth}
  \centering
  \centerline{\includegraphics[height=1.3cm,trim=0 6.375in 9.75in 0,clip]{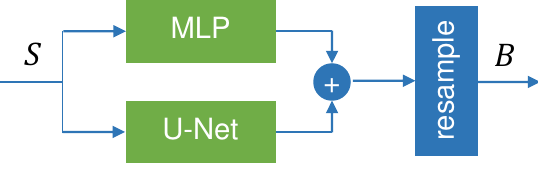}}
  \centerline{(a) Pre-processor}\medskip
\end{minipage}
\hspace{1mm}
\begin{minipage}[b]{0.48\linewidth}
  \centering
  \centerline{\includegraphics[height=1.3cm,trim=0 6.375in 9.75in 0,clip]{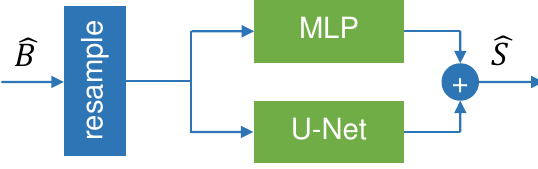}}
  \centerline{(b) Post-processor}\medskip
\end{minipage}
\vspace{-4mm}
\caption{Neural pre-processor and post-processor.}
\label{fig:processor}
\vspace{-3mm}
\end{figure}

\autoref{fig:architecture}(b) shows the setup for training the neural pre-processor and post-processor using stochastic gradient descent.
Because derivatives cannot be back-propagated through the standard image codec, it
is replaced by a differentiable\footnote{In this paper as in most of the ML literature, the term {\em differentiable} more properly means {\em almost-everywhere differentiable}.} \emph{image codec proxy}.  For each training example $n=1,\ldots,N$, the image codec proxy reads the bottleneck image $B_n$ and outputs the reconstructed bottleneck image $\hat{B}_n$, as a standard image codec would.  It also outputs a real-valued estimate of the number of bits $R_n$ that the standard image codec would use to encode $B_n$.
The distortion is measured as any differentiable distortion measure $D_n=d(S_n,\hat{S}_n)$ (such as the squared $\ell_2$ error $||S_n-\hat{S}_n||^2$)
between the original and reconstructed source images.
Together, the rate $R_n$ and distortion~$D_n$ are the key elements of the differentiable loss function.
Specifically, the neural pre-processor and post-processor are optimized to minimize the Lagrangian $D+\lambda R$ of the average distortion $D=(1/N)\sum_n D_n$ and the average rate $R=(1/N)\sum_n R_n$.

\begin{figure}%
\footnotesize
\begin{minipage}[b]{1.0\linewidth}
  \centering
  \centerline{\includegraphics[width=8.5cm, trim=0 400 430 0cm, clip]{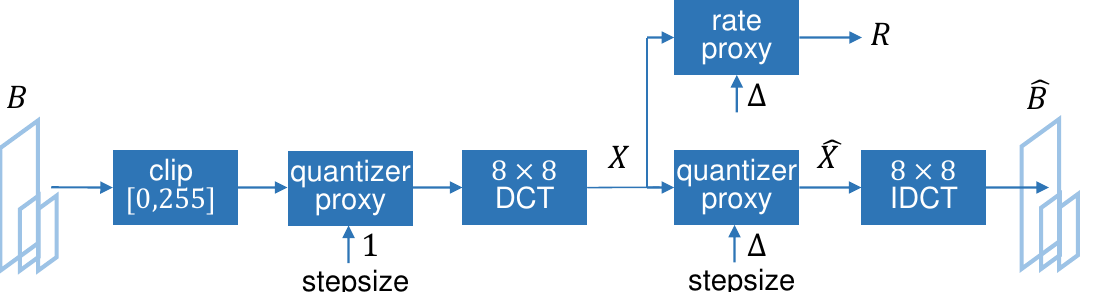}}
\end{minipage}
\vspace{-14pt}
\caption{Image codec proxy.}
\vspace{-5mm}
\label{fig:codec_proxy}
\end{figure}

The image codec proxy itself comprises the differentiable elements shown in \autoref{fig:codec_proxy}.
For convenience the image codec proxy is modeled after JPEG, an early codec for natural images.
Nevertheless, experimental results show that it induces the trained pre-processor and post-processor to produce bottleneck images sufficiently like natural images that they can also be compressed efficiently by other codecs such as HEVC (or VVC/AV1, see \cite{yueyu_icip}.)  The image codec proxy spatially partitions the input channels
into $8\times8$ blocks.
In the DCT domain, the blocks $X=[X_i]$ are processed independently, using
(1)~a ``differentiable quantizer'' (or {\em quantizer proxy}) to create distorted DCT coefficients~$\hat{X}_i=Q(X_i)$, and
(2)~a differentiable entropy 
measure (or {\em rate proxy}) to estimate the bitrate required to represent the distorted coefficients $\hat{X}_i$.
Both proxies take the nominal quantization stepsize $\Delta$ as an additional input.
Further information on quantizer and rate proxies is provided in supplementary %
section~\ref{sec:supp-q-and-r-proxies},
their adaptations for HR and HDR are provided in  
\ref{sec:supp-hr-and-hdr-adaptations}.
\vspace{-5mm}

\subsection{Sandwich for Video Compression}
\vspace{-1mm}

\begin{figure*}[t]
    \centering
    \includegraphics[width=0.8\linewidth,trim=0in 1.5in 0in 0.25in,clip]{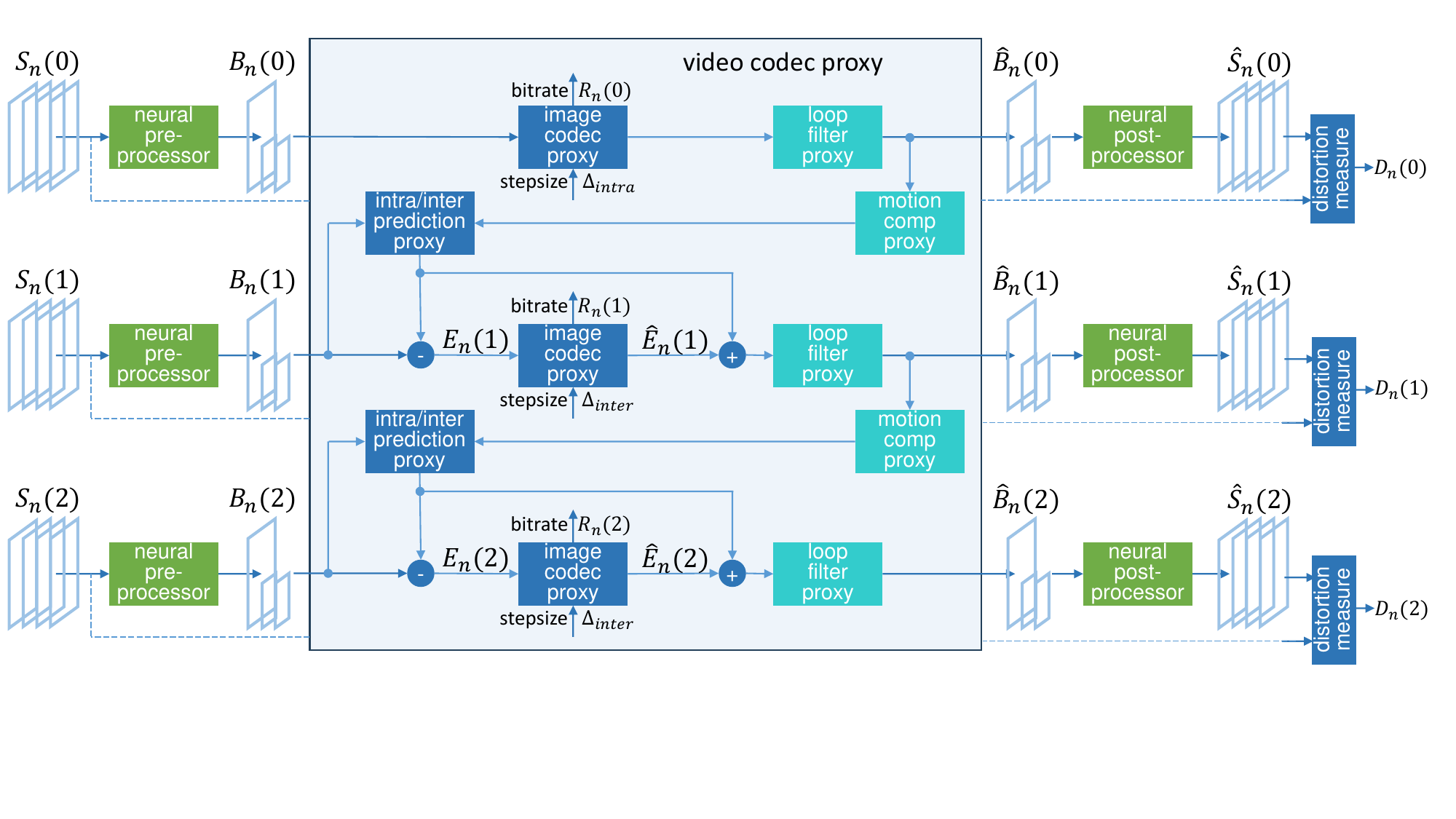}
    \caption{Neural-sandwiched video codec during training. Loss function for training is $\sum D_n(t)+\lambda R_n(t)$ over example clips $n$ and frames $t$.  The shaded box (video codec proxy) is replaced with a standard video codec during operation/inference.  Green boxes are trainable; blue are differentiable; cyan are differentiable with pre-trained parameters.  The entire video codec proxy is differentiable.
    }
    \label{fig:diagram}
    \vspace{-5mm}
\end{figure*}

The sandwich architecture for video compression is shown wrapping our video codec proxy in \autoref{fig:diagram}. Observe that the neural pre-post-processors handle video frames independently enabling straightforward spatio-temporal parallelization.
The video codec proxy maps an input video sequence to an output video sequence plus a bitrate for each video frame.  It has the following components.

{\em Intra-Frame Compression}.  The video codec proxy simulates coding the first ($t=0$) frame of the group, or the I-frame, using the image codec proxy described above in \autoref{sec:image_sandwich}.  %

{\em Motion Compensation}.  The video codec proxy simulates predicting each subsequent ($t>0$) frame of the group, or P-frame, by motion-compensating the previous frame.
Motion compensation is performed using a 
pre-computed
dense motion flow field obtained by running a state-of-the-art optical flow estimator, UFlow \cite{lodha1996uflow}, between the {\em original} source images $S_n(t)$ and $S_n(t-1)$.  The video proxy simply applies this 
motion flow to the previous reconstructed bottleneck image $\hat B_n(t-1)$ to obtain an inter-frame prediction $\tilde B_n(t)$ for bottleneck image $B_n(t)$.  
Note that our motion compensation proxy does not actually depend on $B_n(t)$, so 
even though it is a spatial warping, it is a linear map from $\hat B_n(t-1)$ to $\tilde B_n(t)$, with a constant Jacobian.  This makes optimization much easier than 
using
a differentiable function of both $\hat B_n(t-1)$ and $B_n(t)$
that {\em finds} as well as {\em applies} a warping from $\hat B_n(t-1)$ to $B_n(t)$.  Such functions have notoriously fluctuating Jacobians that make training difficult.

{\em Prediction Mode Selection}.  
An Intra/Inter prediction proxy simulates a modern video codec's Inter/Intra prediction mode decisions.
This ensures better handling of temporally occluded/uncovered regions in video. First, Intra prediction is simulated by rudimentarily compressing the current-frame and low pass filtering it. This simulates filtering, albeit not the usual directional filtering, to predict each block from its neighboring blocks.  Initial rudimentary compression ensures that the Intra prediction proxy is not unduly preferred at very low rates. For each block, the Intra prediction (from spatial filtering) is compared to the Inter prediction (from motion compensation), and the one closest to the input block determines the mode of the prediction.

{\em Residual Compression}.  The predicted image, comprising a combination of Intra- and Inter-predicted blocks, is subtracted from the bottleneck image, to form a prediction residual.  The residual image is then 
compressed using the image codec proxy described above in \autoref{sec:image_sandwich}.
The compressed residual is added back to the prediction to obtain a ``pre-loop-filtered'' reconstruction of the bottleneck image $\hat B_n(t)$.

{\em Loop Filtering}.  The ``pre-loop-filtered'' reconstruction is then filtered by a loop filter to obtain the final reconstructed bottleneck image $\hat B_n(t)$.  The loop filter is implemented with a small U-Net((8);(8, 8)) \cite{UNet:15} (see \autoref{sec:complexity} for U-Net notation) that processes one channel at a time.  The loop filter is trained once for four rate points on natural video using only the video codec proxy with rate-distortion
training loss in order to mimic common loop filters. The resulting set of filters are kept fixed for all of our simulations, \ie, once independently trained, the four loop filters 
are not further trained.

{\em Pre-Post-Processors.} Same as \autoref{sec:image_sandwich} and   \autoref{fig:processor}.

{\em Loss Function}.
The loss function is the total rate-distortion Lagrangian $\sum_{n,t}D_n(t)+\lambda R_n(t)$,
where $D_n(t)$ and $R_n(t)$ are the distortion and rate of frame $t$ in clip $n$. 
The rate term 
serves to encourage the pre- and post-processors to produce temporally consistent neural codes, since neural codes that move according to 
the
motion 
field are well predicted (see \autoref{sec:video-400}.)
Note that the overall mapping from the input images through the pre-processor, video codec proxy, post-processor, and loss function is differentiable.
\vspace{-2mm}

\section{Image Compression Experiments}
\label{sec:images}
\vspace{-2mm}
We first present results for compressing ordinary 3-channel color (RGB) images through codecs with a restricted number of channels (4:0:0 and 4:2:0 compared to 4:4:4), where distortion is measured as RGB-PSNR.
Then we present results for compressing high spatial resolution (HR) %
images through codecs with lower spatial resolution (LR), and for compressing high dynamic range (HDR) images through codecs with lower dynamic range (LDR), where distortion is again measured as RGB-PSNR.
These results are indicative of how a neural sandwich can adapt the hardware of a standard codec to source images with higher resource requirements.
Finally, we present results 
on graphics data, first
for compressing 3-channel normal maps, where distortion is measured as PSNR on the normal maps, and 
then for compressing 8-channel shading maps for use in computer graphics, where distortion is measured as RGB-PSNR {\em after shading} from 8 to 3 channels and across a multitude of views.
These results are indicative of how a neural sandwich can adapt a codec designed for 3-channel RGB images and MSE to other image types
and other distortion measures.
In these results, we also see how the codec proxy, though modeled after JPEG, is also adequate to represent more advanced codecs such as HEIC/HEVC-Intra.

\begin{figure*}[t]
\begin{minipage}{0.32\linewidth}
  \centering
  \includegraphics[width=1.0\linewidth, trim=11mm 0mm 15mm 6mm, clip]{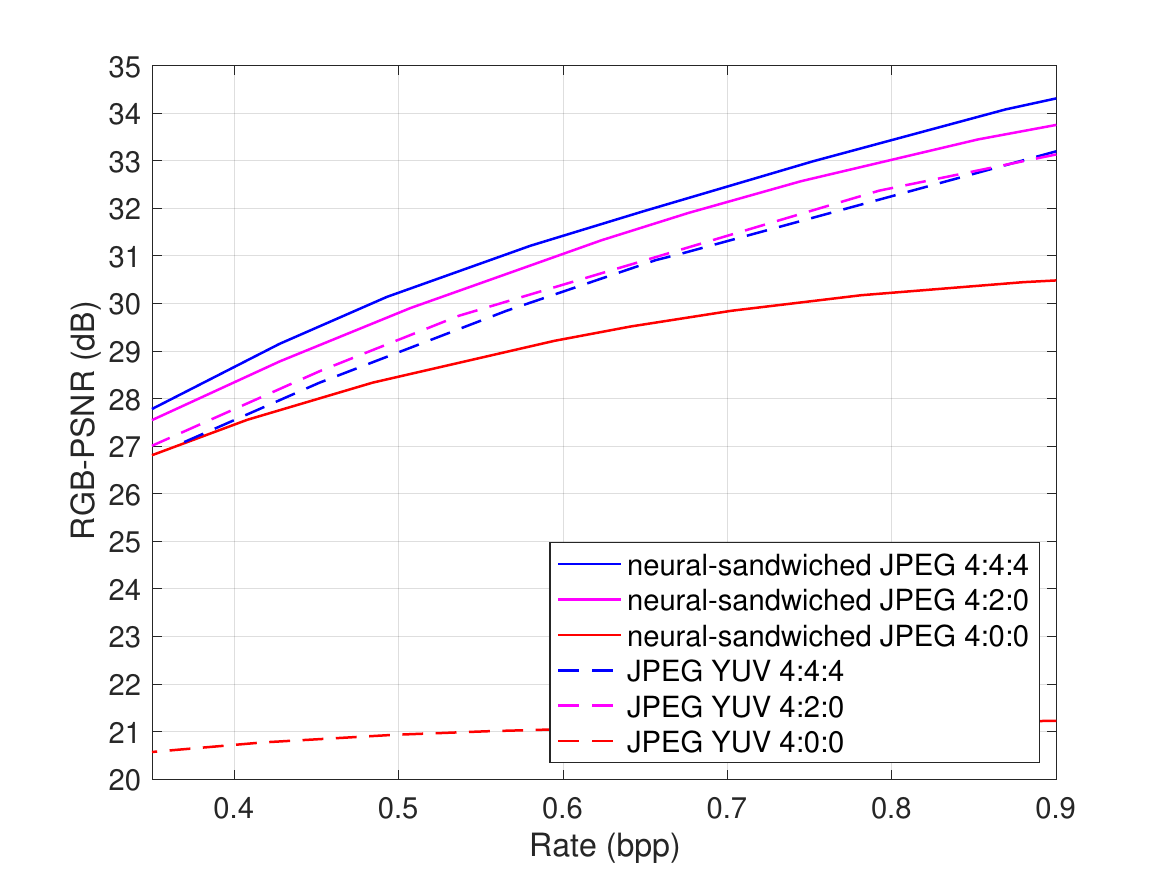}
  \caption{R-D performance of compressing 3-channel RGB images with JPEG YUV and neural-sandwiched JPEG, in 4:4:4, 4:2:0, and 4:0:0.
  }
  \label{fig:jpeg_image_results_rgb}
\end{minipage} \hspace{2mm}
\begin{minipage}{0.64\linewidth}
  \begin{minipage}{0.49\linewidth}
  \centering
  \centerline{\includegraphics[width=1.0\linewidth, trim=11mm 2mm 15mm 8mm, clip]{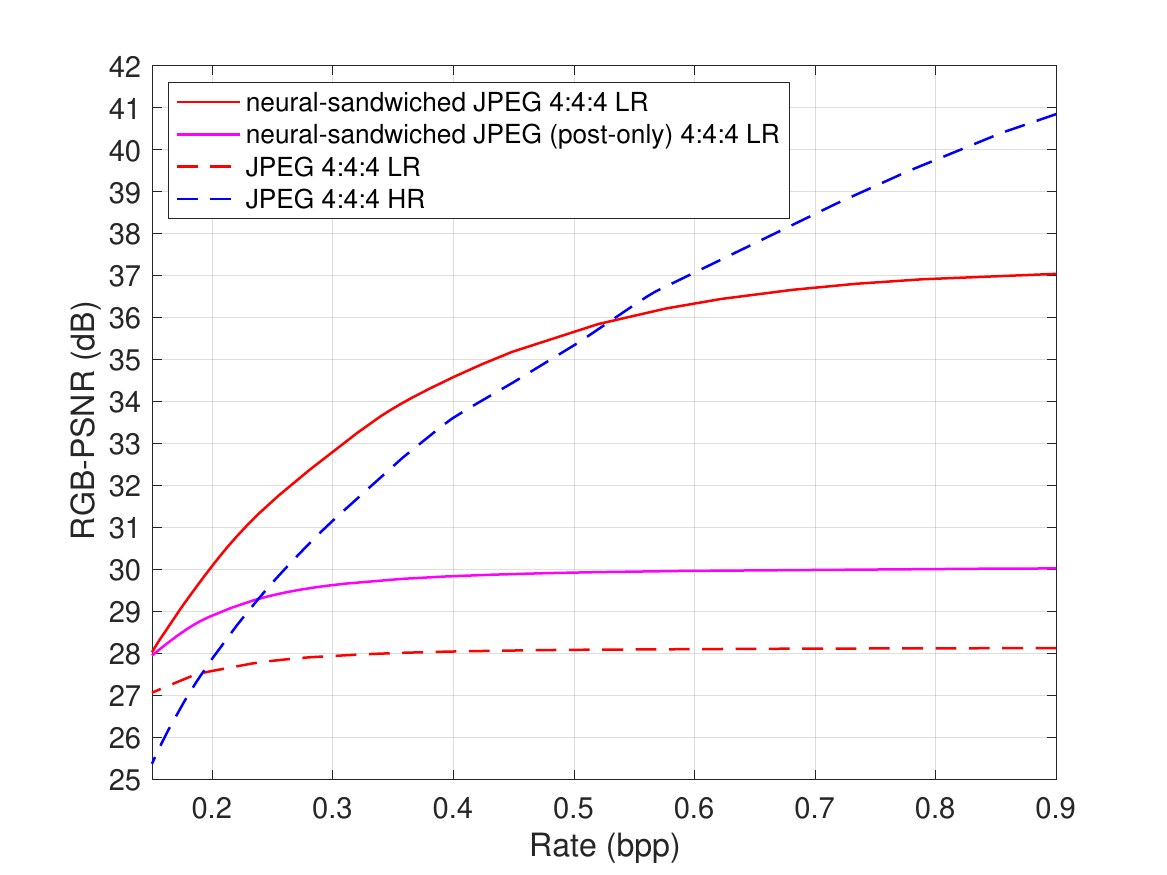}}  %
  \centerline{\footnotesize (a) JPEG}
  \end{minipage}
  \begin{minipage}{0.5\linewidth}
  \includegraphics[width=1.0\linewidth, trim=11mm 2mm 15mm 8mm, clip]{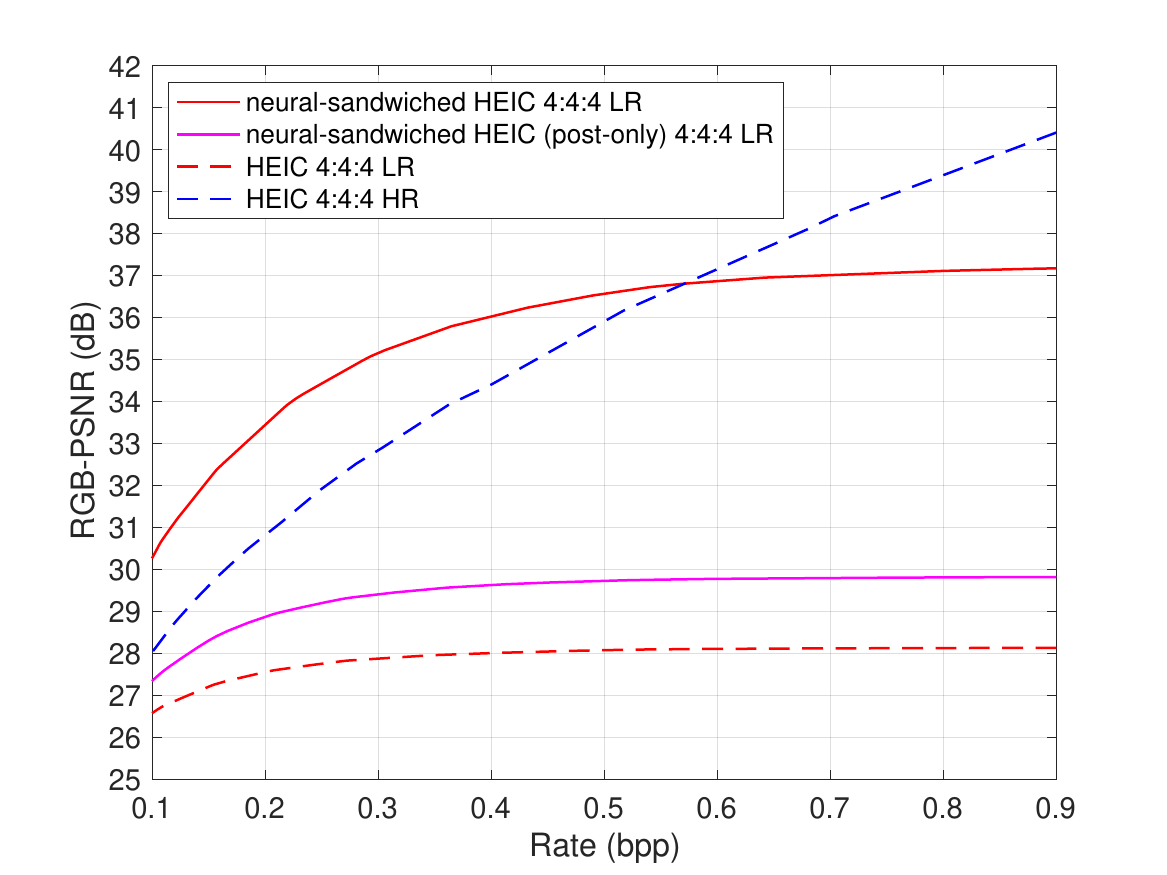}
  \centerline{\footnotesize (b) HEIC}
  \end{minipage}
  \caption{R-D performances of compressing high resolution (HR) RGB images with lower resolution (LR) codecs alone, LR codecs plus neural post-processing, and neural-sandwiched SR codecs.  For reference, also shown are R-D performances of compressing HR RGB images with HR codecs.
  }
  \vspace{-.25cm}
  \label{fig:SR-RD}
\end{minipage}
\vspace{-3mm}
\end{figure*}

For our RGB results, we use the Pets, CLIC, and HDR+ datasets \cite{pet_dataset,clic_dataset, hasinoff2016burst},
while for our computer graphics results, we use the Relightables dataset \cite{relightables}.
Training and evaluation are performed on distinct subsets of each dataset.
All 
results are generated using actual compression on 500 test images of size $256\times256$ randomly cropped from the evaluation subset.

The U-Nets have a multi-resolution ladder of four with channels doubling up the ladder from $32$ to $512$ , specifically U-Net([32, 64, 128, 256]; [512, 256, 128, 64, 32]).  (See \autoref{sec:complexity} for notation.) MLP networks have two layers with $16$ nodes.
Output channels of the networks are determined based on bottleneck and overall output channels.

We obtain R-D curves as follows.
We train four models $m_i$, $i=1,\ldots,4,$ 
for four 
different Lagrange multiplier values $\lambda_i$
using established $D + \lambda R$ optimization \cite{gersho_and_gray}.
For each model, we obtain an R-D curve by encoding the images using a sweep over many step-size values.
Finally we compute the Pareto frontier of these four curves.

\vspace{-3mm}
\subsection{Compressing 
RGB Images with \texorpdfstring{$C\leq3$}{C≤3}-channel Codecs}
\label{sec:rgb-images}
\vspace{-1mm}
In this section, we report the rate-distortion performance of compressing 3-channel RGB images with JPEG YUV and neural-sandwiched JPEG codecs,
across 4:4:4, 4:2:0, and 4:0:0 formats, which respectively correspond to $C$-channel bottleneck images with $C=3$, $1.5$, and $1$.

\autoref{fig:jpeg_image_results_rgb} shows R-D results 
evaluated on the Pets dataset.
For the 4:0:0 format, the neural-sandwiched version shows 6--9 dB improvement over the standard codec, %
due to the fact that for this format, the standard codec can transport only grayscale, whereas the neural-sandwiched version manages to transport color through modulating patterns, as exemplified in \autoref{fig:modulation_dithering1}.
For the 4:2:0 and 4:4:4 formats, R-D performances for the standard codec 
are close to one another.
For both formats, 
the neural-sandwiched standard codec performs better than the standard codec 
in either format.  In particular, for the 4:2:0 format, the neural-sandwiched version shows a 10\% reduction in bitrate, while for the 4:4:4 format, the neural-sandwiched version shows a 15\% reduction. 
It is interesting to see that,
unlike the case for the standard codec, %
that 
{\em the neural-sandwiched version finds a way to utilize the denser sampling of the 4:4:4 format for improved gains over 4:2:0}.
Also surprising is that at low rates {\em the neural-sandwiched 4:0:0 codec becomes competitive with the standard
codecs.}
\vspace{-4mm}

\subsection{Compressing High Resolution (HR) RGB Images with Lower Resolution (LR) Codecs}
\label{sec:hr_lr_images}
\vspace{-1mm}
In this subsection, we
consider the rate-distortion performance of compressing high resolution (HR) RGB images with lower resolution (LR) standard codecs, both with JPEG and HEIC. We also contrast results with ``CNN-RD'' \cite{EusebioAP20}.
LR is half the horizontal and vertical resolution of HR.
Thus, whether sandwiched or not, we precede the standard codecs with bicubic filtering and $2\times$ downsampling and follow them with $2\times$ upsampling using Lanczos3 interpolation.

\autoref{fig:SR-RD} shows R-D results using (a) JPEG and (b) HEIC as the standard codec
evaluated on the CLIC dataset.  When they are neural-sandwiched, JPEG and HEIC have identical pre- and post-processors, with no retraining for HEIC.  The neural-sandwiched standard codecs show substantial improvement over the standard codecs alone: 5-7 dB gains for JPEG, and 7-8 dB gains for HEIC, in the 0.3-0.5 bpp range, and further gains at higher rates.
We also compare a standard codec with a post-processor alone (\ie, with no pre-processor), where the post-processor is architecturally identical
to the post-processor in the neural sandwich, but trained to perform only super-resolution. 
It can be seen that the post-processor alone accounts for at most 2 dB of the sandwich's gains.  The substantial improvement obtained by the sandwich over the super-resolution network clearly points to {\em the importance of the neural pre-processor, the joint training of the pre- and post-processor networks}, and their ability to communicate with each other using neural codes to signal how to super-resolve the images.
For reference, the figure also shows R-D performances of compressing HR RGB images with codecs that are natively HR.  It can be seen that the sandwiched LR codecs outperform even the native HR codecs over a wide range of lower bitrates.  Comparing JPEG and HEIC results, it can be seen that the gains due to neural sandwiching 
one
are substantially retained for 
the other.

\autoref{fig:SR_images} shows examples that are instructive in understanding the advantages of the sandwich over post-processing alone.
Observe the substantial improvements obtained by the sandwiched codec over JPEG and neural post-processing: Detail is retained in the city view, %
aliasing is avoided on the building %
face
and the texture, and text is visible on the keyboard. All with substantial dB improvements (+4.5 dB, +6.5 dB, +4.5 dB over neural post-processing) at the same rate. Interestingly note the very patterned and noisy looking bottlenecks as depicted in \autoref{fig:SR_bottleneck_images1} in supplementary \autoref{sec:supp-hr_lr_images}.

\begin{figure}
  \includegraphics[width=0.3\linewidth, trim=10mm 0mm 10mm 20mm, clip]{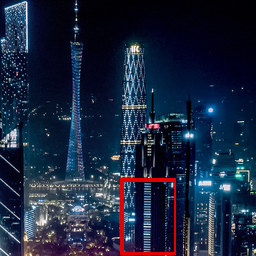}
  \includegraphics[width=0.3\linewidth, trim=10mm 0mm 10mm 20mm, clip]{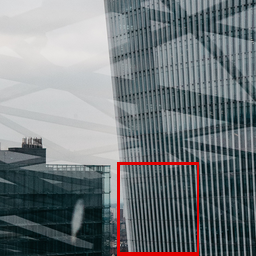}
  \includegraphics[width=0.3\linewidth, trim=10mm 0mm 10mm 20mm, clip]{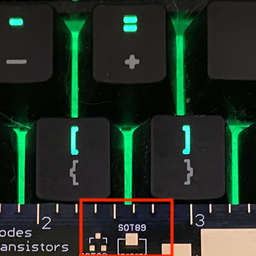}\vspace{-1.2mm}
  \centerline{\footnotesize {\bf (a)} Originals}\vspace{1mm}
  
  \includegraphics[width=0.3\linewidth, trim=10mm 0mm 10mm 20mm, clip]{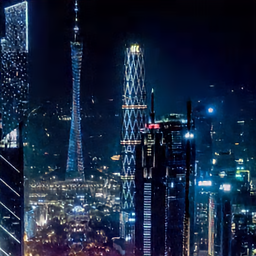}
  \includegraphics[width=0.3\linewidth, trim=10mm 0mm 10mm 20mm, clip]{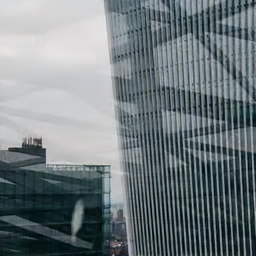}
  \includegraphics[width=0.3\linewidth, trim=10mm 0mm 10mm 20mm, clip]{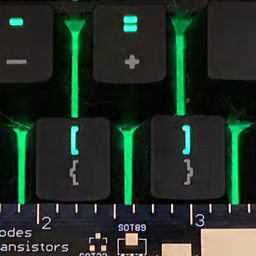}\vspace{-1.2mm}
  \centerline{\footnotesize {\bf (b)} {\em Sandwich:} (29.1 dB, 0.54 bpp),
  (32.1 dB, 0.33 bpp), (35.1 dB, 0.38 bpp)}\vspace{1mm}
  
  \includegraphics[width=0.3\linewidth, trim=10mm 0mm 10mm 20mm, clip]{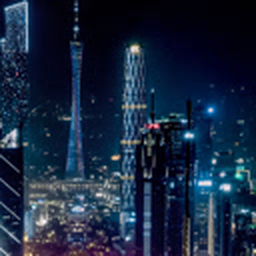}
  \includegraphics[width=0.3\linewidth, trim=10mm 0mm 10mm 20mm, clip]{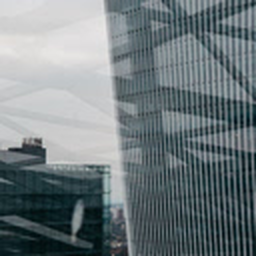}
  \includegraphics[width=0.3\linewidth, trim=10mm 0mm 10mm 20mm, clip]{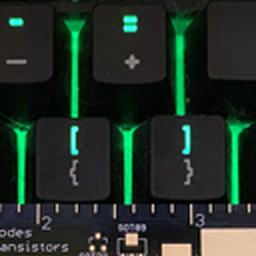}\vspace{-1.2mm}
  \centerline{\footnotesize {\bf (c)} {\em JPEG:} (23.4 dB, 0.54 bpp),
  (23.9 dB, 0.34 bpp), (26.6 dB, 0.38 bpp)}\vspace{1mm}
  
  \includegraphics[width=0.3\linewidth, trim=10mm 0mm 10mm 20mm, clip]{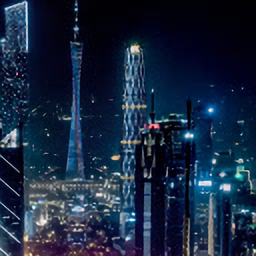}
  \includegraphics[width=0.3\linewidth, trim=10mm 0mm 10mm 20mm, clip]{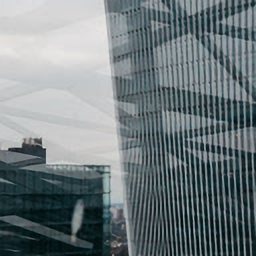}
  \includegraphics[width=0.3\linewidth, trim=10mm 0mm 10mm 20mm, clip]{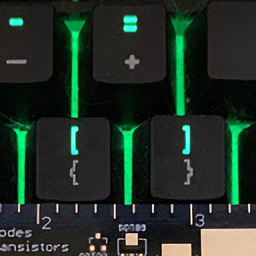}\vspace{-1.2mm}
  \centerline{\footnotesize {\bf (d)} {\em Post-Only:} (24.6 dB, 0.54 bpp),
  (25.6 dB, 0.34 bpp), (30.6 dB, 0.38 bpp)
  
  }
  
  \caption{Super-resolution sandwich of a low-res codec: Original $256\times256$ source images and reconstructions by sandwich, JPEG with linear upsampling, and JPEG with neural post-processing respectively. The regions identified in the top row show areas where detail is either lost 
  or aliased after  downsampling and LR transport. 
  Note how the sandwich output in (b) correctly transports the detail whereas JPEG and post-only recover the wrong information. The picture in the last column, while correctly transported by the sandwich, results in severe aliasing for JPEG and even further reduced performance for post-only which amplifies the aliasing. 
  The interested reader can generate an extensive set of further examples using our software at \cite{sandwich_oss}.
  }
  \label{fig:SR_images}
  \vspace{-5mm}
\end{figure}

\autoref{tab:CNN-RD-comparison} compares our HR sandwich to the closely related but independently developed work of ``CNN-RD'' \cite{EusebioAP20}, which 
also surrounds a standard codec with neural pre- and post-processors using 2x down- and up-sampling.  However their networks'
formulation and training regimen prohibits them from learning to communicate the neural codes needed to carry good HR information.  (Their post-processor is trained first to super-resolve a low-pass image; then their pre-processor is trained to minimize $D+\lambda R$ with the fixed post-processor.  This misses the main advantage of having neural pre- and post-processors.)  The table shows that 
our work has
significantly higher gains in PSNR-Y (dB) relative to the same standard codec (JPEG) on the Div2k validation image 0873 \cite{Agustsson_2017_CVPR_Workshops}.  Indeed, though not shown in the table, their solution saturates and begins to under-perform the standard codec above 0.8 bpp ($\sim$30 dB); ours out-performs the standard codec until about 2.0 bpp ($\sim$37 dB).
\vspace{-3mm}

\begin{table}[t]
    \centering
    \begin{tabular}{c|ccccccc}
    bpp & 0.2 & 0.3 & 0.4 & 0.5 & 0.6 & 0.7 & 0.8 \\ \hline
    CNN-RD\cite{EusebioAP20} & 1.58 & 1.09 & 0.67 & 0.55 & 0.33 & 0.18 & 0.15 \\
    HR sandwich & \textbf{1.59} & \textbf{1.49} & \textbf{1.42} & \textbf{1.49} & \textbf{1.49} & \textbf{1.46} & \textbf{1.69}
    \end{tabular}
    \caption{Gain in PSNR-Y (dB) over JPEG on Div2k validation image 0873.}
    \label{tab:CNN-RD-comparison}
    \vspace{-7mm}
\end{table}

\subsection{Compressing HDR RGB Images with LDR Codecs}

In this subsection, we report the rate-distortion performance of compressing high dynamic range (HDR) RGB images with lower dynamic range (LDR) standard codecs (8-bit JPEG and 8-bit HEIC), alone,  neural-sandwiched, and compared to Dequantization-Net \cite{liu2020single}.%
For the HDR simulations, we use the HDR+ dataset \cite{hasinoff2016burst}.
The HDR images have $d=16$ bits per color component, while the LDR codecs transmit only $8$ bits per color component.
By sandwiching the LDR codecs in 
a sandwich,
it is possible to signal spatially-localized tone mapping curves via neural codes from the pre-processor to the post-processor, in order to carry the least significant bits of the HDR image through the LDR codec.

\begin{figure*}[t]
\begin{minipage}{0.3\linewidth}
  \centerline{\includegraphics[width=1.0\linewidth, trim=11mm 0mm 15mm 8mm, clip]{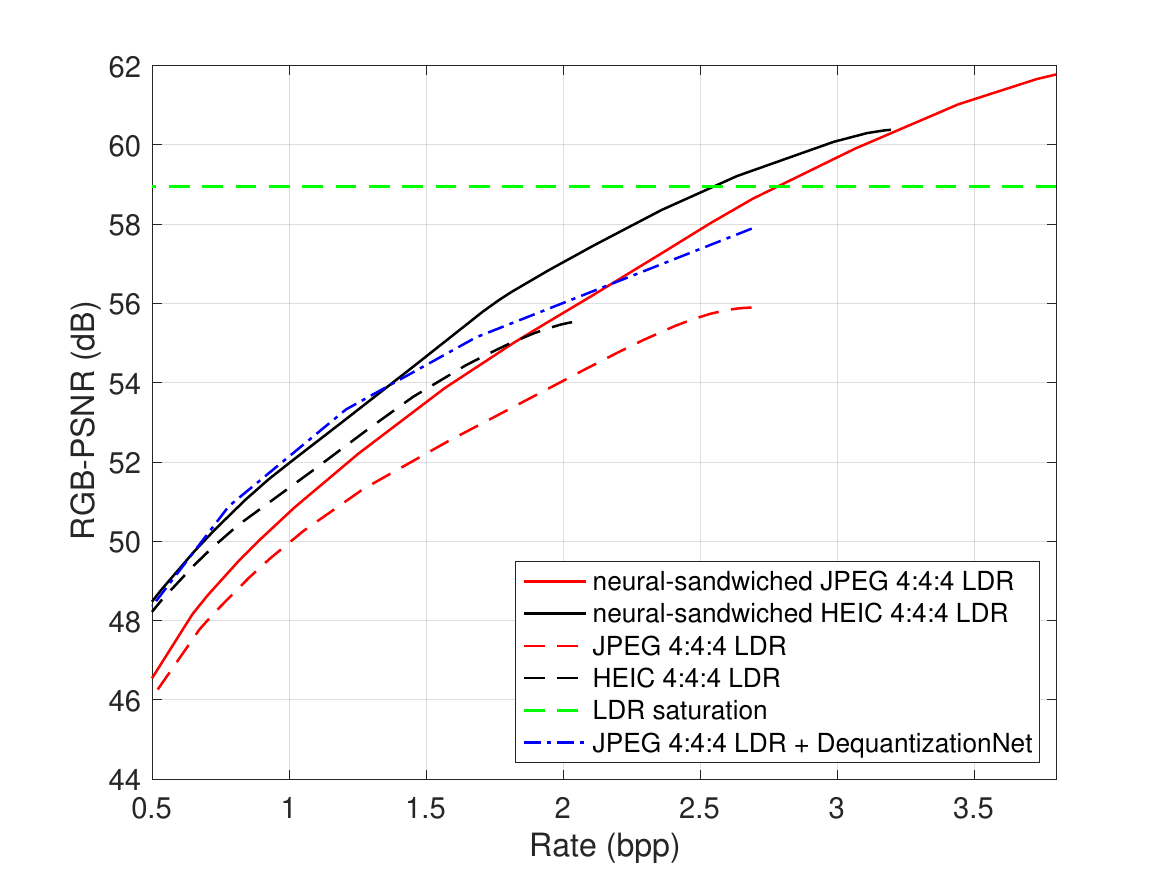}}
  \caption{R-D performances of compressing HDR 
  images
    with LDR, neural-sandwiched , and Dequantization-Net \cite{liu2020single} codecs.
  }
  \label{fig:HDR-RD}
\end{minipage} \hspace{2mm}
\begin{minipage}{0.3\linewidth}
\centering
  \includegraphics[width=1.0\linewidth, trim=11mm 3mm 15mm 8mm, clip]{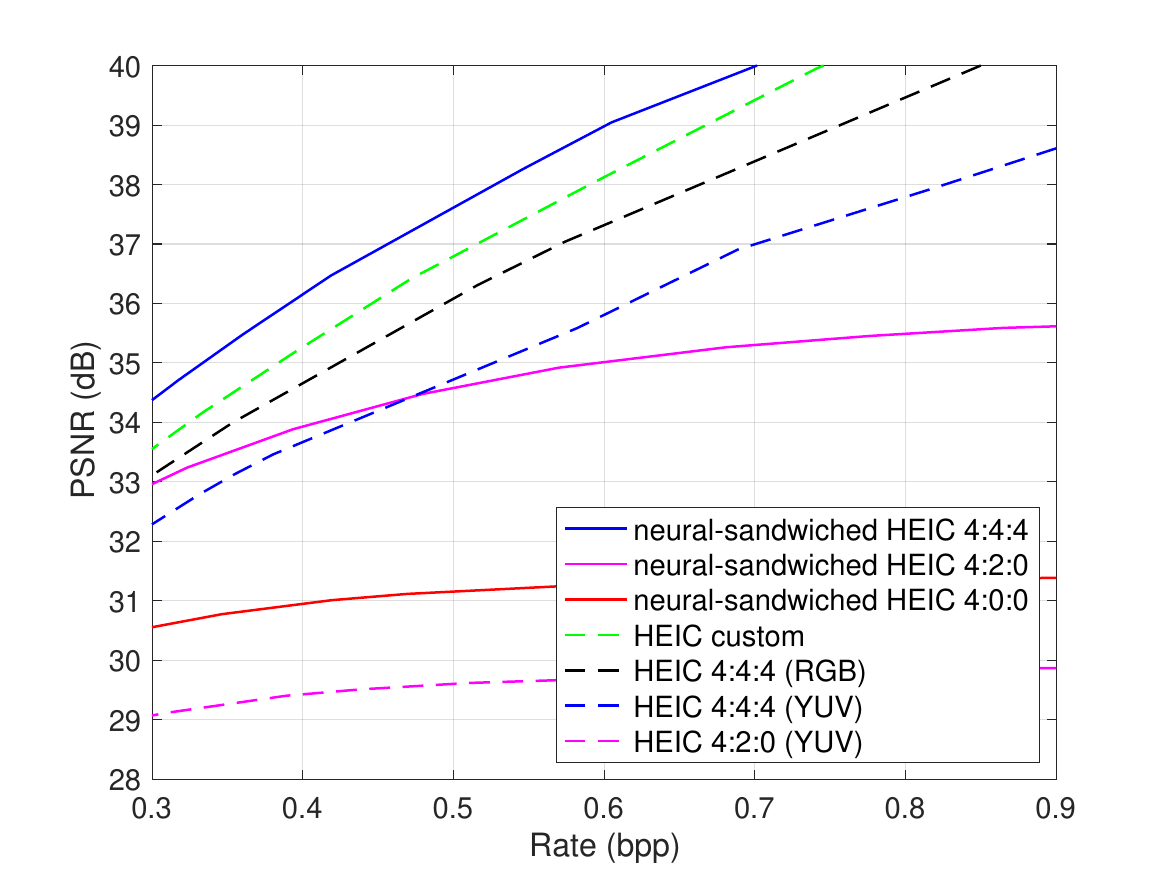}
  \caption{R-D performances of compressing normal map images with HEIC and neural-sandwiched HEIC, in various formats.
  }
  \label{fig:results_normals1}
\end{minipage} \hspace{2mm}
\begin{minipage}{0.3\linewidth}
  \centering
  \includegraphics[width=1.0\linewidth,trim=11mm 3mm 15mm 8mm,clip]{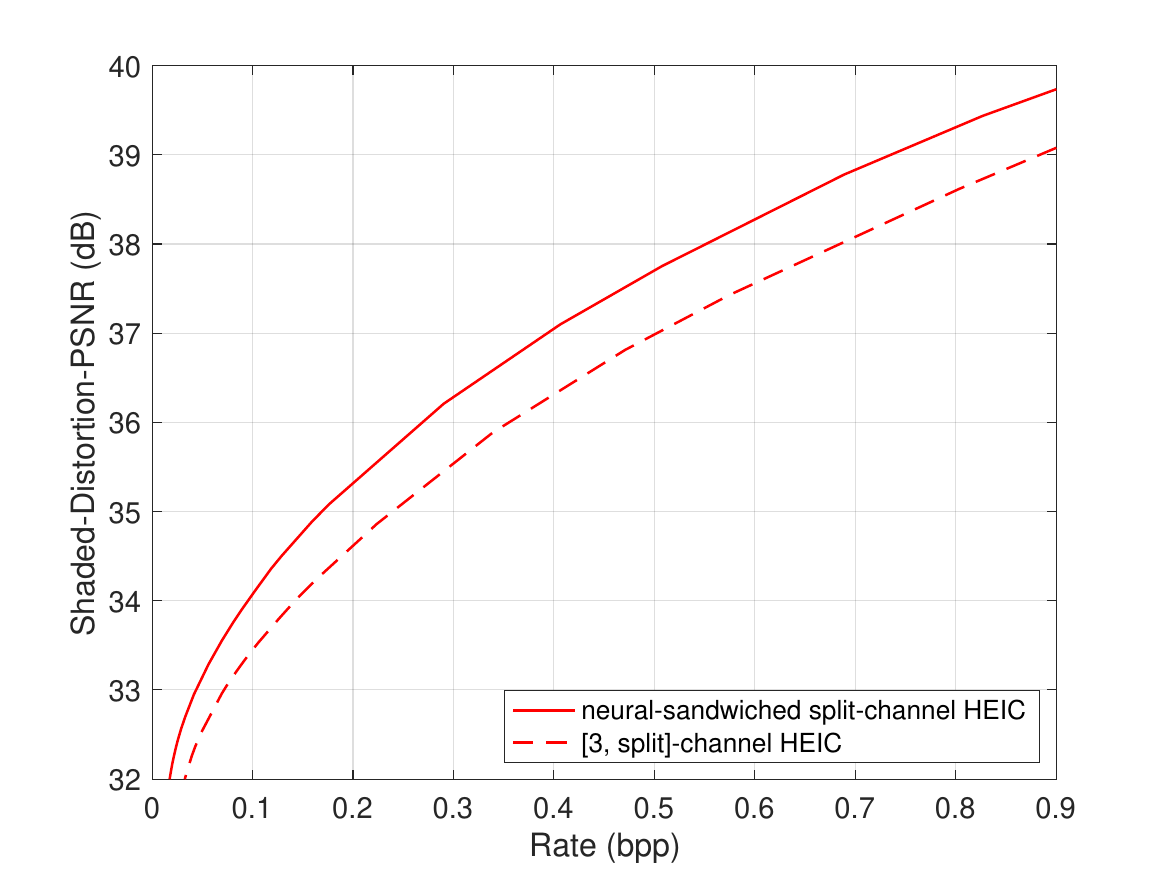}
  \caption{R-D performance of compressing 8-channel texture map images with standard codec alone and neural-sandwiched.
  }
  \label{fig:texture_map_results}
\end{minipage} \hspace{2mm}
\vspace{-5mm}
\end{figure*}
\autoref{fig:HDR-RD} shows R-D results in terms of $d$-bit RGB-PSNR vs.\ bits per pixel, evaluated on the HDR+ dataset.
The %
$d$-bit RGB-PSNR is given by
\vspace{-2mm}
\begin{equation}
    10\log_{10}\left(\left( 2^d-1 \right)^2 (3 H W) / {
    \left\Vert S-\hat S \right\Vert ^2 }\right) ,
    \vspace{-2mm}
\end{equation}
where $S$ and $\hat{S}$ are the original and reproduced $d$-bit RGB source images of size $H\times W$.
The figure shows the performance of the LDR codecs alone (8-bit JPEG and 8-bit HEIC) in comparison to the neural-sandwiched LDR codecs, as well as to JPEG post-processed with the state-of-the-art Dequantization-Net \cite{liu2020single} (trained on the same dataset). The maximum PSNR one can obtain by losslessly encoding the most significant $8$-bits is illustrated as LDR saturation. The standard codecs alone, or with the Dequantization-Net post-processor only, saturate at that level. Observe that the sandwiched codecs rise up to 3 dB above the saturation line, highlighting the importance of joint training of the pre- and post-processors and communication between them using neural codes. Unfortunately the software implementing the standard codecs precludes the transmission of higher rates. Neither our JPEG nor HEIC implementation is able to go beyond $\sim$3 bpp on average. For all R-D curves the highest rate point is where the software cuts off. Using codec implementations accomplishing higher rates, the gains of the sandwich are expected to increase further.
\vspace{-2mm}

\vspace{-3mm}
\subsection{Compressing 3-channel Normal Maps with Color Codecs}
\label{sec:normal_maps}

Many computer graphics applications periodically transport normal maps to GPUs in compressed form to decompress and accomplish sophisticated lighting effects that increase
the perceived resolution of a mesh \cite{enwiki:1191014419, BeersAC:96}.
An example normal map is shown in
\autoref{fig:texture_map_images} (b).
Each pixel stores a unit-norm vector $n=(n_x, n_y, n_z)$ representing
the tangent-space normal ($n_z\geq 0$) with respect to a 
mesh.
Note that three channels are redundant since $n_x^2+n_y^2+n_z^2=1$.
To represent these as 8-bit images, 
each channel $[-1,1]$ is mapped to RGB $[0,255]$.

\autoref{fig:results_normals1} show R-D results using HEIC 
evaluated on normal maps from the Relightables dataset \cite{relightables}.
The best results are obtained with the neural-sandwiched 4:4:4 codecs, which exploit the channel redundancy.
Neural-sandwiching codecs 
provide significant gains over their respective baselines.
YUV 4:2:0 and YUV 4:4:4 codecs use RGB-YUV conversion, which does not provide any advantage for this dataset.
Codecs with 4:4:4 (no color conv.) are better.
The best non-neural result (custom) is obtained
by zeroing out the third component $n_z$ during compression,
and recovering it in a postprocess as $\hat{n}_z = 1 - (\hat{n}_x^2 + \hat{n}_y^2)^{\frac{1}{2}}$.  %
However, the best result overall (by more than 1 dB) uses a neural-sandwiched codec with a 4:4:4 format.
Comparing HEIC results to the JPEG case (see \autoref{fig:results_normals2} for JPEG), it can again be seen that the gains due to neural sandwiching of JPEG are substantially retained for HEIC, with about 15\% reduction in bitrate compared to HEIC alone.
\vspace{-3mm}

\vspace{-1mm}
\subsection{Compressing Computer Graphics having 8-channel Texture Maps with 3-channel Color Codecs and Shaded Distortion}
\vspace{-1mm}
\begin{figure}[b]
\vspace{-4mm}
\begin{minipage}{1.0\linewidth}
\footnotesize
\centering
\begin{minipage}{0.45\linewidth}
  \centering
  \centerline{\includegraphics[width=4.3cm,trim=0cm 9.3cm 0cm 0cm,clip]{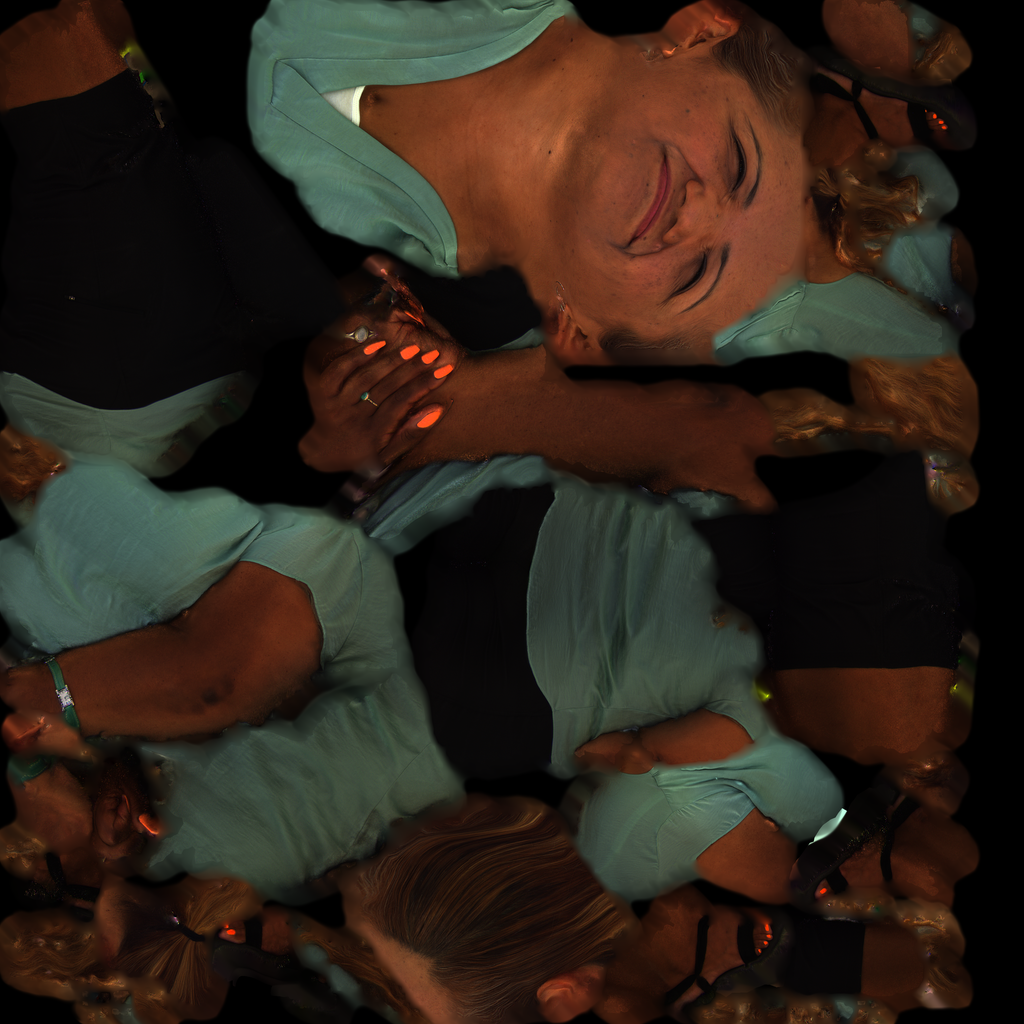}}
  \centerline{(a) 3-channel albedo}
  \vspace*{1.5mm}
\end{minipage}
\hspace{2mm}
\begin{minipage}{0.45\linewidth}
  \centering
  \centerline{\includegraphics[width=4.3cm,trim=0cm 9.3cm 0cm 0cm,clip]{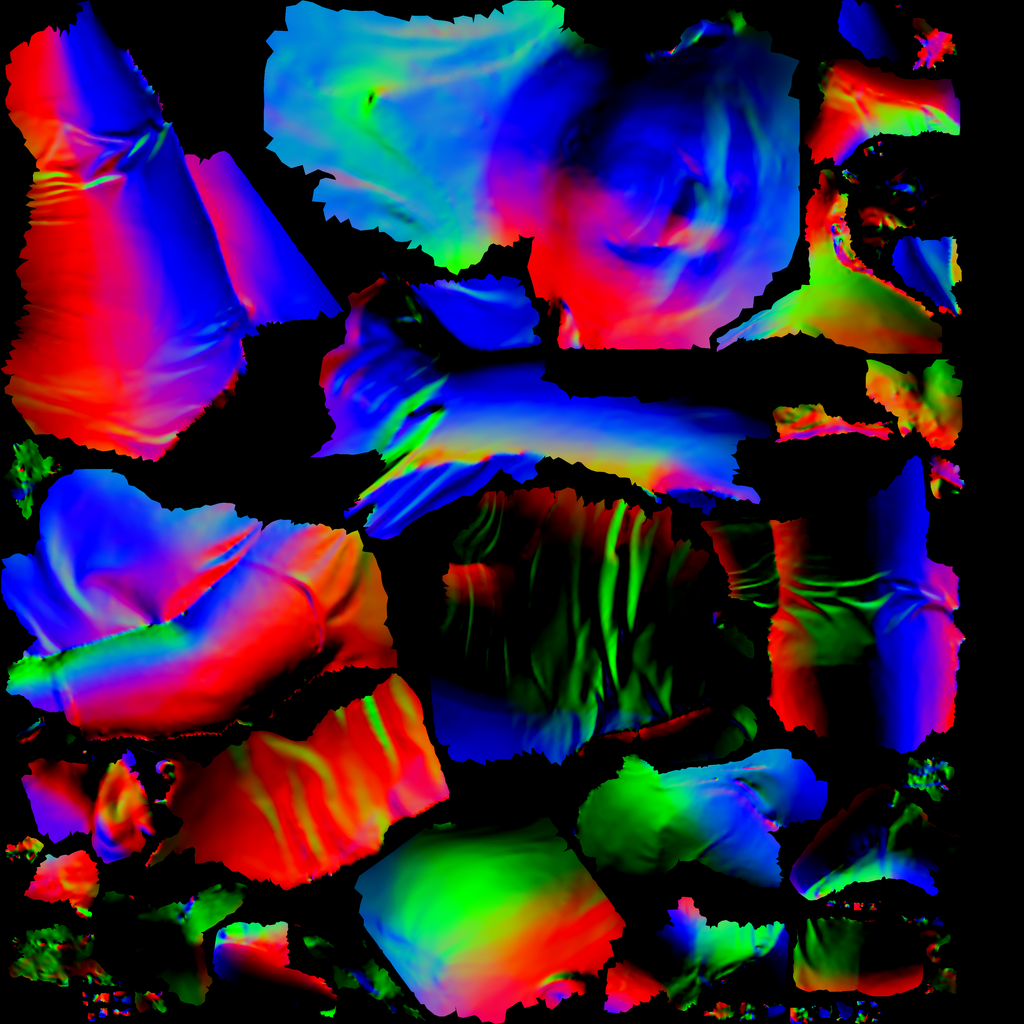}}
  \centerline{(b) 3-channel normals}
  \vspace*{1.5mm}
\end{minipage}
\begin{minipage}{0.45\linewidth}
  \centering
  \centerline{\includegraphics[width=4.3cm,trim=0cm 9.3cm 0cm 0cm,clip]{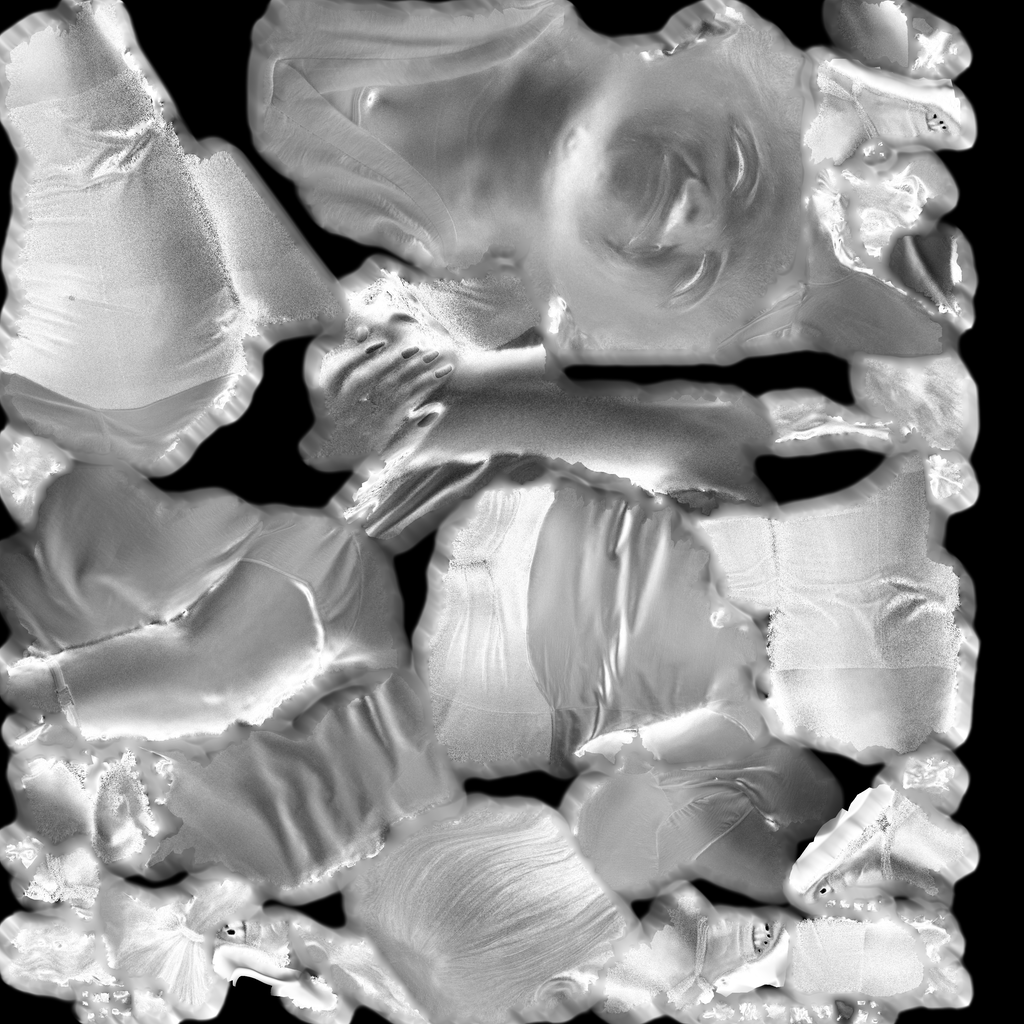}}
  \centerline{(c) 1-channel roughness~}
  \vspace*{1.5mm}
\end{minipage}
\hspace{2mm}
\begin{minipage}{0.45\linewidth}
  \centering
  \centerline{\includegraphics[width=4.3cm,trim=0cm 9.3cm 0cm 0cm,clip]{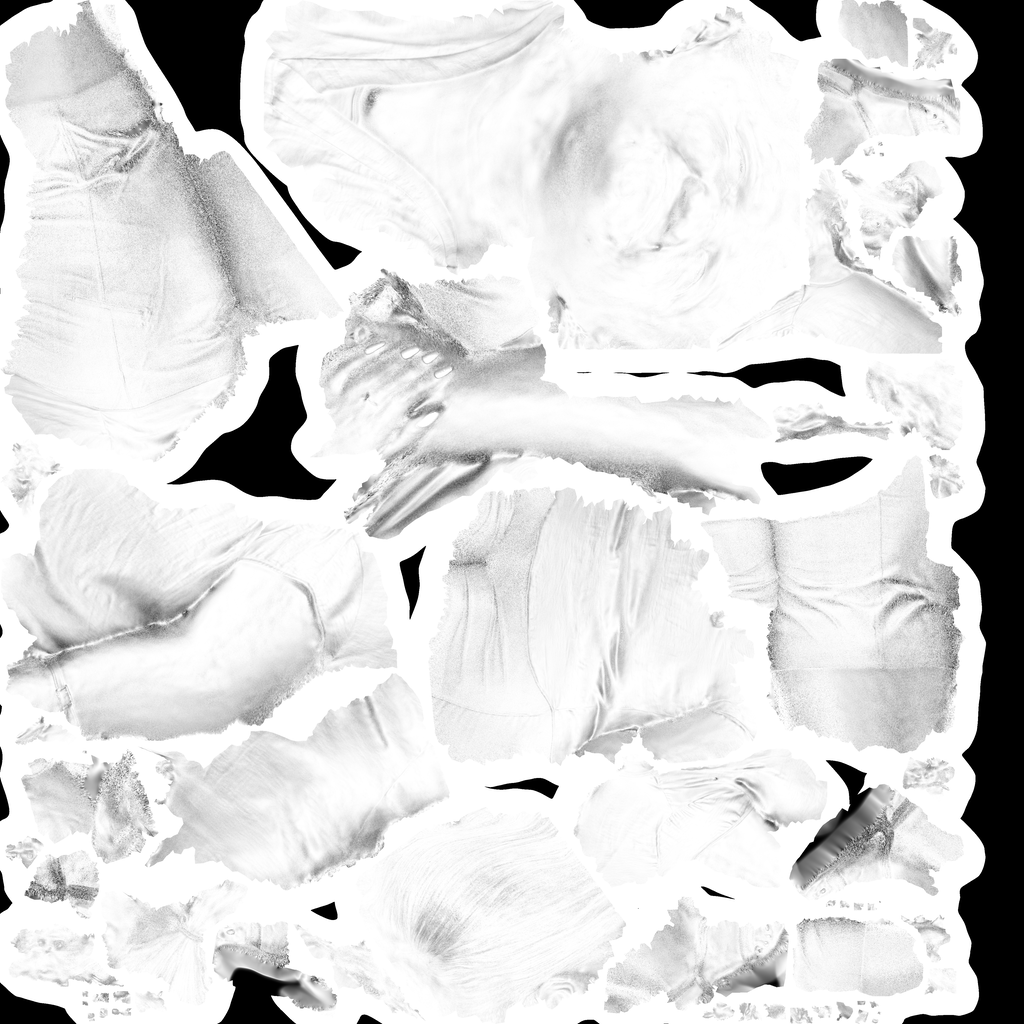}}
  \centerline{(d) 1-channel occlusion}
  \vspace*{1.5mm}
\end{minipage}
\hfill
\begin{minipage}{0.95\linewidth}
  \centerline{
  \includegraphics[width=2.83cm,trim=0cm 9.3cm 0cm 0cm,clip]{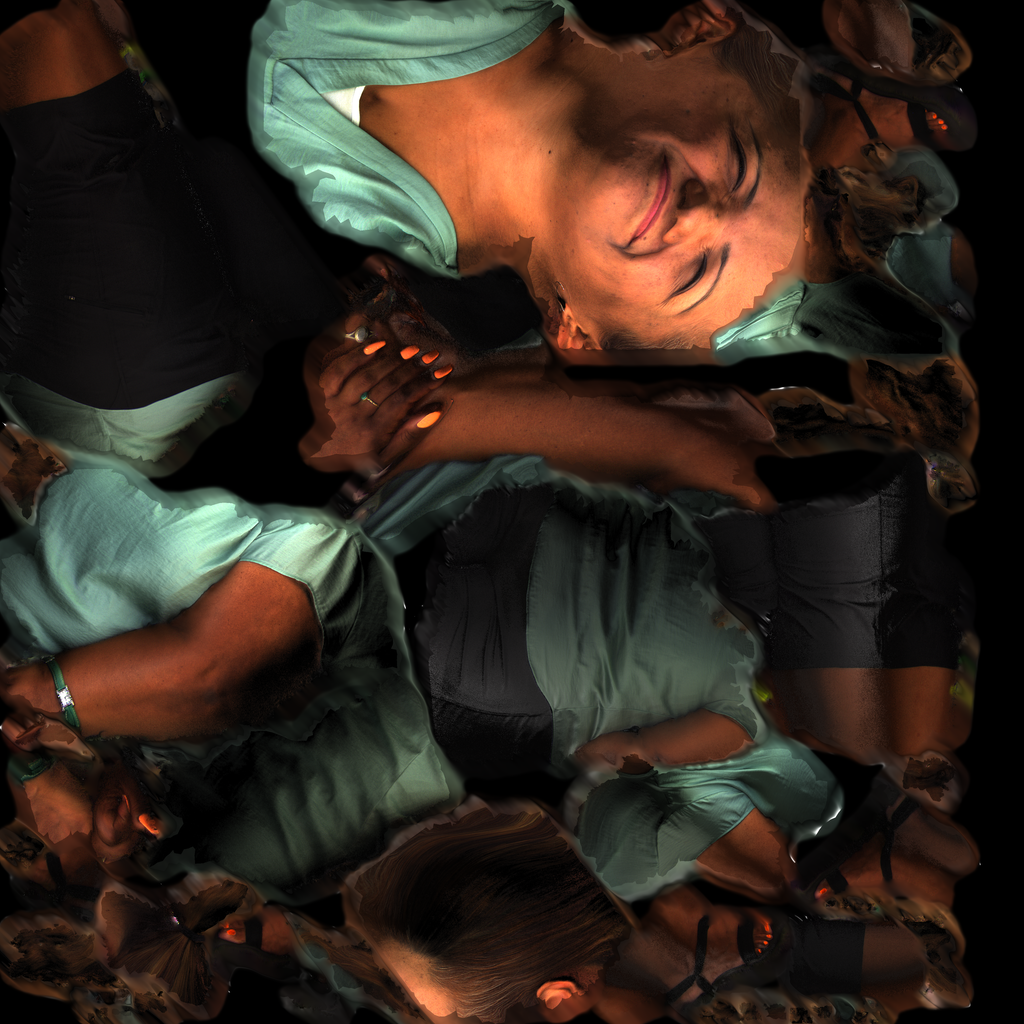}\hspace{0.5mm}%
  \includegraphics[width=2.83cm,trim=0cm 9.3cm 0cm 0cm,clip]{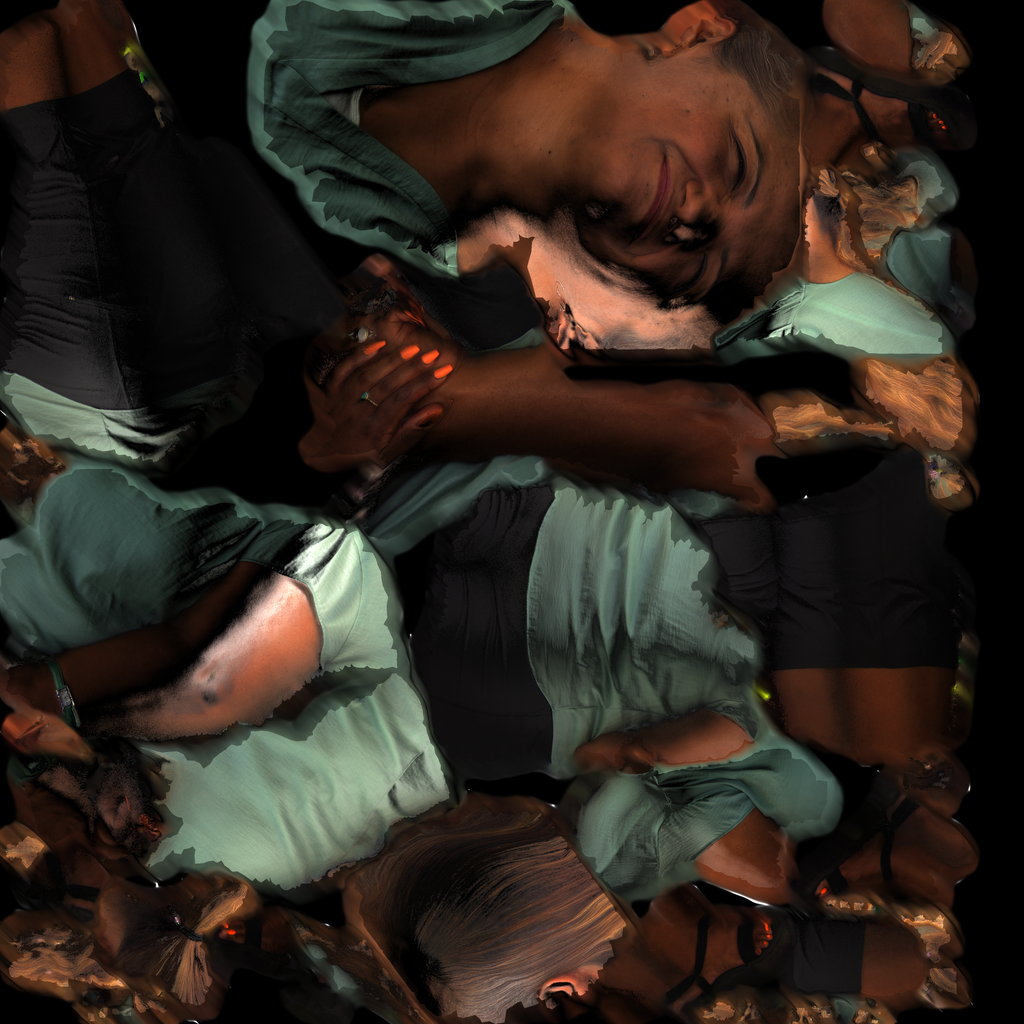}\hspace{0.5mm}%
  \includegraphics[width=2.83cm,trim=0cm 9.3cm 0cm 0cm,clip]{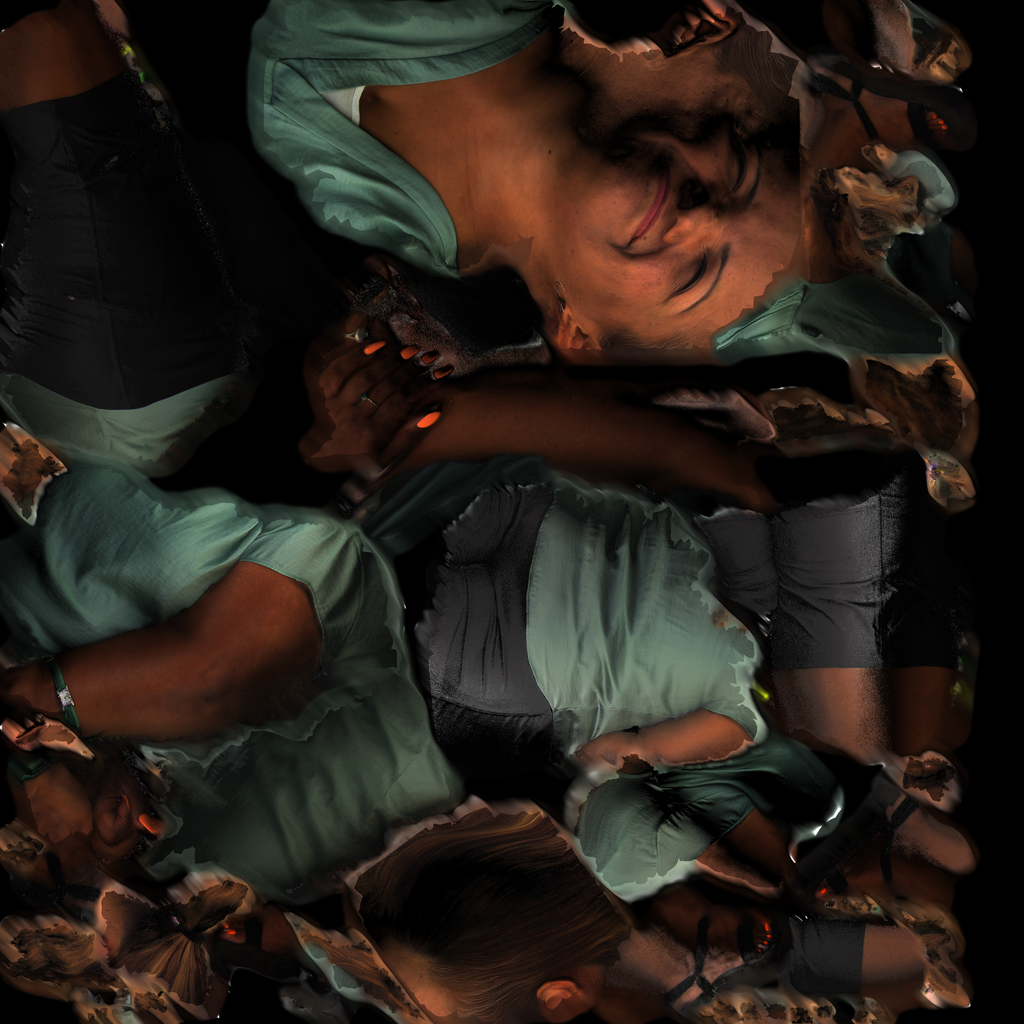}
  }
  \centerline{(e) Shaded using different viewpoints and lights}
\end{minipage}
\caption{Components of 8-channel texture map and samples of images rendered to measure shaded distortion.}
\label{fig:texture_map_images}
\vspace{-2mm}
\end{minipage} 
\end{figure}
A common technique in computer graphics is to render a surface using \emph{texture mapping},
which stores sampled surface properties in an associated texture atlas image (see \autoref{fig:texture_map_images}).
Texture maps comprise large graphics assets and are repeatedly sent to the GPU in compressed form as rendered scenes vary.
The texture map often contains not just albedo (RGB color)
but additional surface properties (\eg, surface normals, roughness, ambient occlusion)
that enable more realistic shading.
In many graphics applications texture assets are compressed with standard codecs and transported prior to final rendering.
In this section, we report rate-distortion performance of compressing $8$-channel texture maps with and without sandwiching.

We briefly review the rendering process.
To render a surface mesh from a particular view and with particular lighting,
rasterization identifies the screen-space pixels covered by the mesh triangles.
For each pixel, it obtains the interpolated 3D surface position
as well as interpolated texture coordinates.
Then, a \emph{pixel shader} computes the view direction, the light direction(s),
and the sampled texture attributes for the surface point.
The final shaded RGB color for the pixel is a complicated formula involving all of these inputs.

For compression in this scenario,
it is natural to {\em measure distortion not over the texture image values
but over the final rendered pixel colors}.
Specifically, we measure the average RGB MSE of images rendered using
a collection of typical views and lighting conditions.
We call this \emph{shaded distortion}.

In principle, it should be possible to train the neural sandwich end-to-end
by measuring shaded distortion over rendered images using a \emph{differentiable renderer}.
However, this proves challenging for the large texture map sizes (up to 4K) encountered in practice.
Instead, 
we adopt a novel approach and
measure shaded distortion in the
domain of the texture images.
That is, we compute shading using all the 3D parameters of the traditional rendering,
but output the resulting shaded colors onto an image defined over the texture domain itself
(see \autoref{fig:texture_map_images}e).
The key benefit is that the computation is local, so training can use cropped texture images.
(For final evaluation, we measure shaded distortion over traditionally rendered images.)

In our experiments, we use a texture map with $C=8$ channels: 3 RGB albedo channels, 3 normal map channels, 
1 roughness channel, and 1 occlusion channel, as illustrated in \autoref{fig:texture_map_images}.  We refer to this as a $[3,3,1,1]$-channel texture map.

\autoref{fig:texture_map_results} shows R-D results in terms of PSNR of the shaded distortion in dB vs bits per texture map pixel.
Compression with and without neural sandwiching are compared.
Without sandwiching, to compress the 8-channel texture map with a standard codec, we partition the texture map into its natural components, here with $[3,3,1,1]$ channels, and compress each component separately (with one codec using 4:4:4 format and color conversion, one codec using 4:4:4 format and no color conversion, and two codecs using 4:0:0 format).
With sandwiching, to compress the 8-channel texture map, for concreteness we choose 8-channel bottleneck images at the same resolution as the texture map, and
code each of the 8 bottleneck channels as a grayscale image using HEIC 4:0:0.
It can be seen from the figure that neural-sandwiching provides a 20-30\% reduction in bitrate compared to HEIC alone.
Clearly, neural sandwiching can be used for non-RGB images with $C>3$ channels and non-standard distortion measures.

\begin{figure}[t]
\centering
\begin{minipage}{0.7\linewidth}
  \centering
  \includegraphics
   [width=1.0\linewidth, trim=11mm 0 10mm 10mm, clip]
  {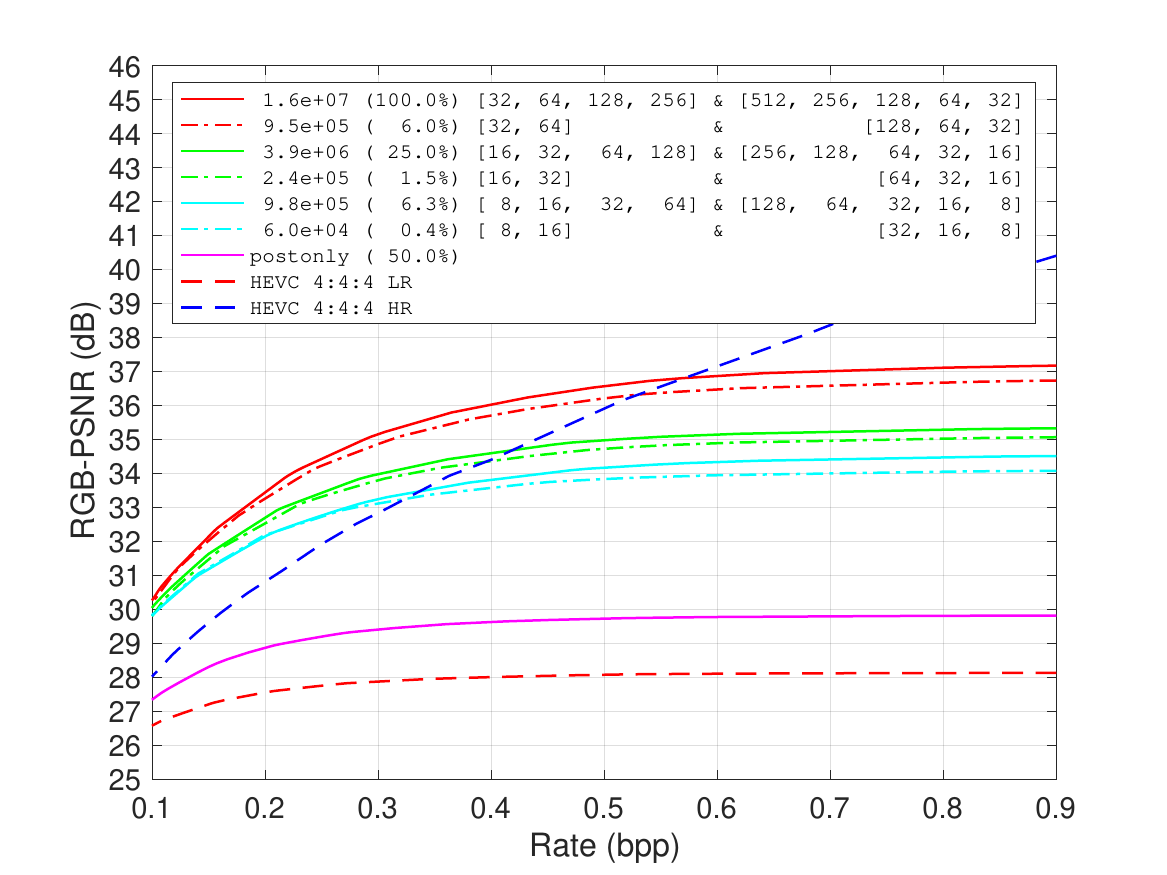}
\end{minipage}
  \caption{R-D performance of compressing HR RGB images using LR codecs, with neural sandwiches of different complexity.  Legend shows parameter counts (for pre-and-post-processor combined) %
  for different U-Nets %
  }
  \label{fig:complexity_results}
  \vspace{-5mm}
\end{figure}

\vspace{-2mm}
\section{Model Complexity}
\label{sec:complexity}
\vspace{-1mm}
In the image compression experiments of \autoref{sec:images}, we use 
an MLP in parallel with a U-Net.  The MLP is relatively simple, with two hidden layers, each with 16 hidden channels.
Most of the complexity is in the U-Net, which in its standard form \cite{UNet:15} has four 2-layer $3\times3$ convolutional ``encoder'' blocks each followed by a $2\times$ downsampling, followed by five 2-layer $3\times3$ convolutional ``decoder'' blocks separated by $2\times$ upsampling and concatenation with the output of the same-resolution encoder block.  The number of channels output from the encoder blocks is [32, 64, 128, 256], while the number of channels output from the decoder blocks is [512, 256, 128, 64, 32], denoted U-Net([32, 64, 128, 256]; [512, 256, 128, 64, 32]).  A final $3\times3$ convolutional layer produces $C_{out}$ output channels.  These tuples, along with $C_{in}$ and $C_{out}$, can be used as hyper-parameters to specify the U-Net.
\autoref{tab:complexity} illustrates the parameter and MAC-based complexity details of the UNet family explored in this paper. 
Run-times in frames-per-second over a single GPU slice are also shown. These run-times can be directly scaled up with multiple slices by spatially tiling the target frames.
\begin{table}[t]
    \centering
    \setlength{\tabcolsep}{3pt}  %
    \begin{tabular}{lr|rrr} \hline
    \multicolumn{2}{c|}{U-Net hyperparameters} & Number of & MACs & FPS\\
    encoder & decoder & parameters & per-pixel & \\ \hline
    [32, 64, 128, 256] & [512, 256, 128, 64, 32] & 7847491 & 213943 & 71 \\ \hline
    [32, 64] & [128, 64, 32] & 472387 & 112531 & 96\\ \hline
    [16, 32, 64, 128] & [256, 128, 64, 32, 16] & 1963043 & 53981 & 145\\ \hline
    [16, 32] & [64, 32, 16] & 118691 & 28619 & 192\\ \hline
    [8, 16, 32, 64] & [128, 64, 32, 16, 8] & 491347 & 13743 & 238\\ \hline
    [8, 16] & [32, 16, 8] & 29971 & 7399 & 323\\ \hline
    [32] & [32, 32] & 57219 & 43347 & 244\\ \hline
    \end{tabular}
    \vspace{1ex}
    \caption{U-Net complexity for various hyperparameters. \protect\\ Fixed hyperparameters include $C_{in}=C_{out}=3$, \protect\\ filters size = $3\times3$, and layers per block = 2. Frames-per-second (FPS) is measured for an  RGB frame ($1024\times 1024$) on a single cloud A100 slice \cite{nvidia_slice}.}
    \label{tab:complexity}
    \vspace{-7mm}
\end{table}

In this section, we 
study the trade-off between the complexity of the sandwich and its rate-distortion performance in
compressing HR RGB images with LR codecs, as detailed in \autoref{sec:hr_lr_images}.
\autoref{fig:complexity_results} shows R-D results for the HR-LR application, for neural sandwiches with different complexities, assuming the pre- and post-processors have equal complexity.  
While our default U-Net([32, 64, 128, 256]; [512, 256, 128, 64, 32]) has almost 8M parameters, many other U-Nets have many fewer parameters.
For example, U-Net([32, 64]; [128, 64, 32]), which has reduced encoder and decoder blocks, two and three respectively, 
has only 491K parameters, a mere 6\% of the parameters of the default U-Net, with less than 0.5 dB loss in R-D performance.
In contrast, reducing complexity by reducing the number of channels, e.g., U-Net([16, 32, 64, 128]; [256, 128, 64, 32, 16]), has a less desirable trade-off.

As illustrated, massive reductions in the number of parameters are possible with little loss in performance especially with U-Net([32]; [32, 32]) having about 57K parameters, i.e., {\em less than 1\% of the parameters of the default U-Net}.  
In the next section, on video experiments, we show that U-Net([32]; [32, 32]), which we call our {\em slim} network, likewise offers orders of magnitude reduction in parameters with little loss in R-D performance and significantly higher fps.
Of course, optimizing the hyper-parameters, such as the number of layers 
or the convolutional filter size, 
considering asymmetric models, using depth-separable convolutions, etc., can significantly improve complexity further.
Alternative model architectures, which are more efficient in terms of MACs, 
are explored in \cite{yueyu_icip}.  
\vspace{-3mm}

\section{Video Compression Experiments}
\label{sec:video}

\subsection{Codec Setup and Dataset}
We generated a video dataset that consists of 10-frame clips of YUV video sequences from the AOM Common Test Conditions \cite{AV2CTC}
and their associated motion flows, calculated using UFlow \cite{lodha1996uflow}. We use a batch size of 8, \ie, 8 video clips in each batch. Each clip is processed during the dataset generation step such that it has 10 frames of size 256x256, selected from video of fps 
20-40. HEVC is implemented using x265 (IPP.., single reference frame, rdoq and loop filter on.) within ffmpeg.\footnote{Observe that the rate for each clip thus reflects (i) an I-frame amortized over $10$ frames and (ii) container format metadata though we tried to minimize the latter. The reader used to video rates where these two factors are amortized over hundreds of frames should expect higher reported rate numbers per-pixel.}
For each considered scenario, the model is trained for 1000 epochs, with a learning rate of $1e^{-4}$, 
and tested on 120 test video clips. 
RD plots are generated similar to \autoref{sec:images}.
We report results in terms of YUV PSNR when using the $\ell_2$ norm and ``LPIPS (RGB) PSNR'' when using LPIPS (refer to \autoref{sec:video-lpips} for LPIPS particulars.) 

\vspace{-3mm}
\subsection{Compressing 
RGB Video with Gray\-scale Codecs}
\label{sec:video-400}
We start by considering the rate-distortion performance of the sandwich system on the toy example of transporting full-color video over a standard codec that can only carry gray-scale video (HEVC 4:0:0.) As we have already seen in \autoref{fig:modulation_dithering1}, the sandwich-introduced modulation patterns are quite pronounced in terms of spatial extent and in terms of spatial frequency for the image compression version of this scenario. The toy example here is hence especially useful for examining (i) the temporal coherence of the sandwich-introduced patterns, (ii) the role of training with and without motion flows, and (iii) understanding the role of the network receptive field
on the final compression scenario.

Video codecs provide the majority of compression gains by exploiting  temporal redundancies via motion correspondences.  Since the sandwich system uses modulation patterns for message passing and is deployed in a frame-independent fashion it is important that the patterns are introduced in a temporally consistent fashion that can be taken advantage of by standard codecs.
This is clearly visible in \autoref{fig:hevc_400}, which shows three frames of compressed bottlenecks, sandwich-reconstructions, and originals. 
The patterns smoothly move with the scene objects that they are attached to. Through many such visuals we have noted that the networks operate in a translation robust manner with patterns moving with the objects and transitioning over motion/object discontinuities. Note also that unlike the image case the patterns for video are spatially broader which precludes the need for extra accurate motion compensation.

\autoref{fig:rd_400_444} demonstrates the significant improvements (+8.5 dB) that the neural-sandwiched HEVC 4:0:0 obtains over HEVC 4:0:0. As an ablation study we trained a sandwich system using only the image codec proxy for the same scenario, shown in the Figure as ``neural-sandwiched HEVC 4:0:0 (image-proxy)''. 
This system is unaware of the motion compensation process used by the codec and assumes the codec compresses video as a sequence of INTRA frames.
Given a large enough training set, one expects this image-codec-proxy-trained sandwich to achieve translation robustness, as such a dataset depicts similar objects at many different translations. While performing worse than the video-codec-proxy-trained sandwich note that this version is also significantly better than HEVC 4:0:0  ($\sim$8 dB). This confirms the  observation that well-trained pre/post-processors accomplish translation robustness, which is of fundamental importance in the video scenario.

\autoref{fig:rd_400_444} also shows the performance of the ``slim'' simplification, again, significantly outperforming HEVC 4:0:0. Note however that this model performs 3 dB below the full model. Since the slim model is restricted to the finest two resolutions as opposed to the full model's four, its receptive field is significantly smaller than that of the full model. 
This in turn restricts its capacity to deploy spatially large patterns which appear to be advantageous in this problem.
\vspace{-4mm}

\begin{figure*}
\begin{minipage}{0.32\linewidth}
  \centering
  \includegraphics[width=5.8cm, trim=11mm 0 8mm 10mm, clip]{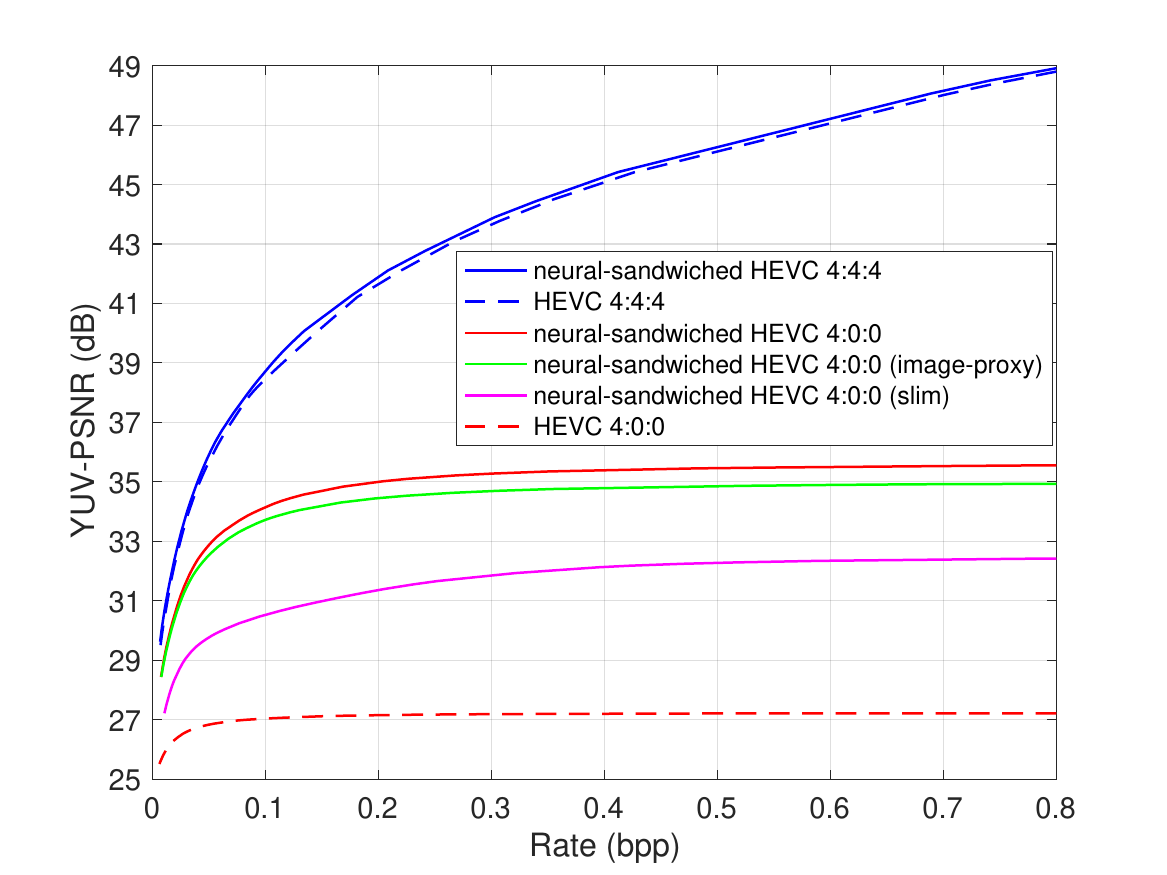}
  \caption{Video rate-distortion performance of the YUV 4:0:0 sandwich and YUV 4:4:4 sandwich.
  }
  \label{fig:rd_400_444}
\end{minipage} \hspace{2mm}
\begin{minipage}{0.32\linewidth}
  \centering
  \includegraphics[width=5.8cm, trim=11mm 0 8mm 10mm, clip]{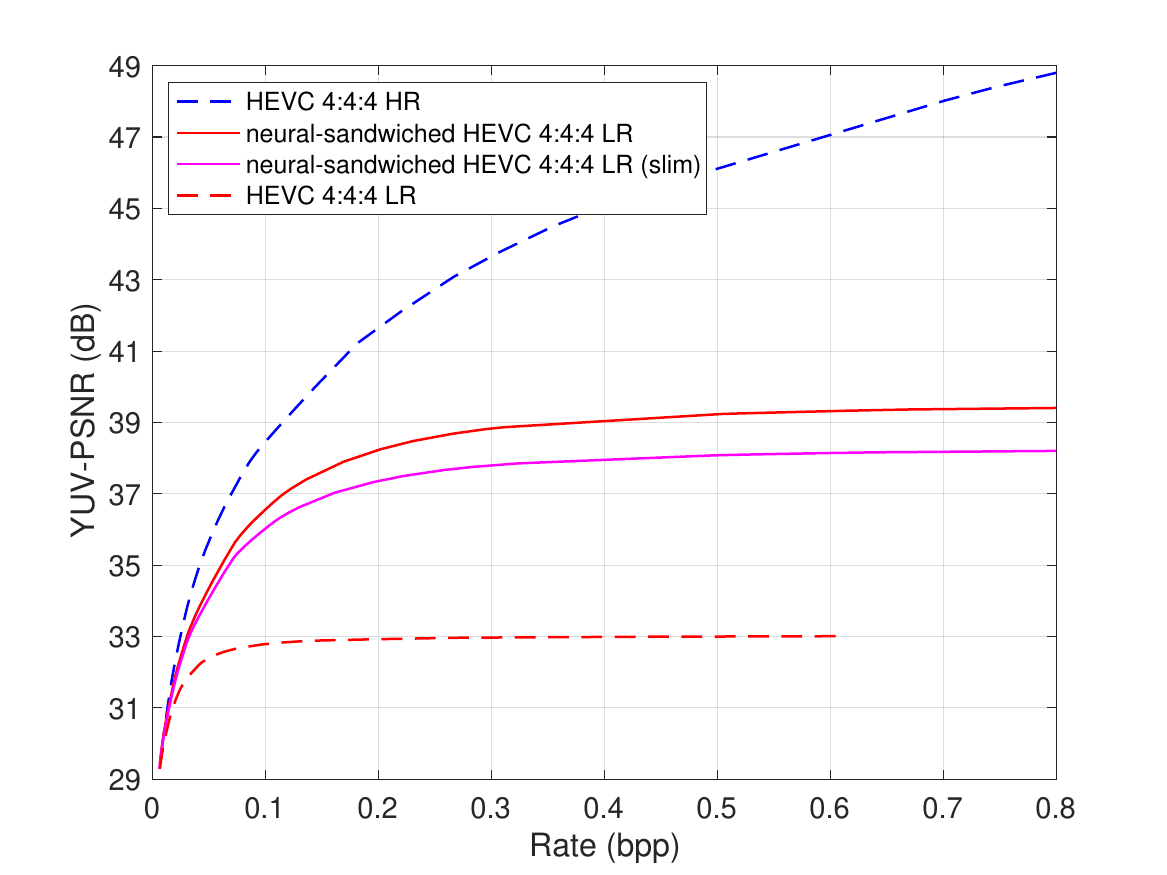}
  \caption{Video rate-distortion performance of the YUV 4:4:4 low-resolution (LR) sandwich.
  }
  \label{fig:video_hrlr_results}
  \end{minipage} \hspace{2mm}
  \begin{minipage}{0.32\linewidth}
  \centering
  \includegraphics[width=5.8cm, trim=11mm 0mm 9mm 10mm, clip]{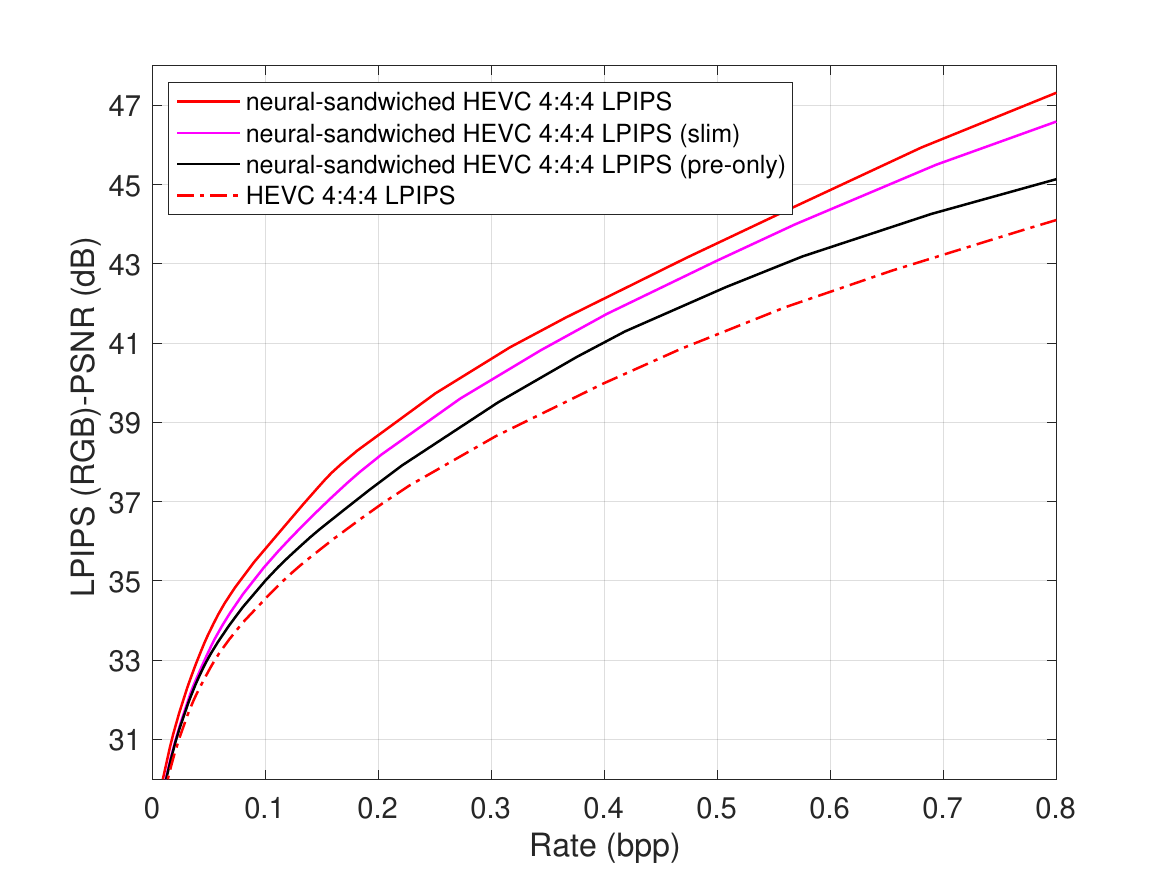}
  \caption{Video rate-distortion performance of sandwich and HEVC with LPIPS.}
  \label{fig:video_lpips_results_rgb}
  \end{minipage}
\vspace{-4mm}
\end{figure*}

\vspace{-2mm}
\subsection{Compressing 3-channel RGB Video with 3-channel Codecs}
\label{sec:rgb-video-444}

The rate-distortion performance of HEVC 4:4:4 and neural-sandwiched HEVC 4:4:4 are also included in \autoref{fig:rd_400_444}. Over a broad rate-range the sandwich readily obtains $\sim$5\% improvements in rate at the same distortion. We have observed that the loop-filter proxy included as part of the video-codec proxy is performing better than the HEVC loop-filter. In effect, rather than compensating for the less-potent HEVC loop filter, the neural-post-processor is trained assuming that the standard codec has a better loop-filter than it actually does. 
This leads to the neural-post-processor leaving some potential post-processing improvements on the table. Adjusting the loop-filter proxy to more closely mimic the HEVC loop filter is expected to marginally improve neural-sandwiched HEVC 4:4:4 results.
\vspace{-5mm}
\subsection{Compressing High Resolution (HR) RGB Video with Lower Resolution (LR) Codecs}
\label{sec:hr_lr_video}
In \autoref{sec:hr_lr_images} we have seen that the sandwich can transport high-resolution (HR) images using lower-resolution (LR) codecs and obtain massive improvements. 
Using both JPEG and HEIC, the sandwich is significantly better over {\em linear-down-codec-linear-up}  and {\em linear-down-codec-neural-up} transport schemes. 

Given the results of \autoref{sec:video-400}, \ie, that the sandwich establishes temporally coherent message passing and continues to obtain massive improvements over the gray-scale codec in the video setting, we expect the sandwich to likewise extend \autoref{sec:hr_lr_images} results to video. Not surprisingly we see this to be the case in  \autoref{fig:video_hrlr_results}, which shows the rate-distortion performance of transporting high-resolution (HR) video using a lower-resolution (LR) codec. Compared to HEVC 4:4:4 LR (Bicubic-down-HEVC4:4:4-Lanczos3-up)  the sandwich obtains more than 6 dB improvements in YUV PSNR. 
The slim model is close with $\sim$5 dB improvements. 
Note the parallels to the image case shown in 
\autoref{fig:complexity_results}.

\begin{figure}[t]
    \centering
\begin{minipage}{0.32\linewidth}
  \centering
  \includegraphics[width=\linewidth, trim=0mm 40mm 40mm 0mm, clip]{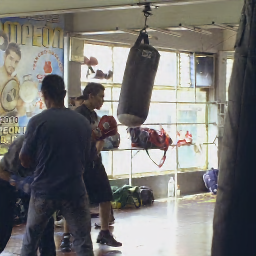}
  \centerline{\footnotesize {\bf (a)} S: 37.0 dB, 0.24 bpp}
\end{minipage} \hfill
\begin{minipage}{0.32\linewidth}
  \centering
  \includegraphics[width=\linewidth, trim=0mm 40mm 40mm 0mm, clip]{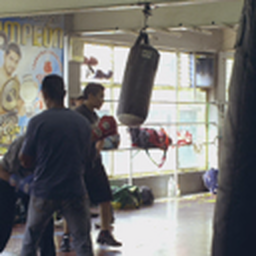}
  \centerline{\footnotesize {\bf (b)} H: 31.1 dB, 0.35 bpp}
\end{minipage} \hfill
\begin{minipage}{0.32\linewidth}
  \centering
  \includegraphics[width=\linewidth, trim=0mm 40mm 40mm 0mm, clip]{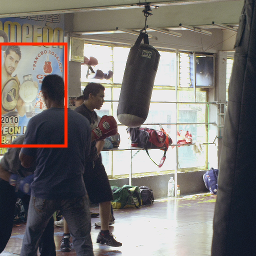} 
  \centerline{\footnotesize {\bf (c)} Original }
\end{minipage} \hfill
\\
\vspace{1mm}
\begin{minipage}{0.32\linewidth}
  \centering
  \includegraphics[width=\linewidth, trim=30mm 0mm 0mm 30mm, clip]{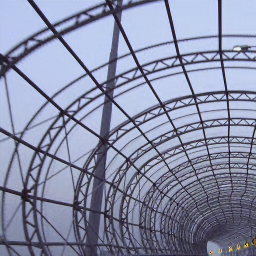}
  \centerline{\footnotesize {\bf (a)} S: 35.1 dB, 0.50 bpp}
\end{minipage} \hfill
\begin{minipage}{0.32\linewidth}
  \centering
  \includegraphics[width=\linewidth, trim=30mm 0mm 0mm 30mm, clip]{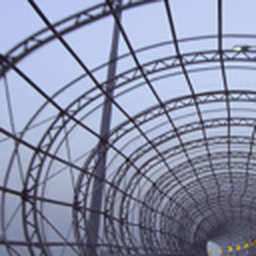}
  \centerline{\footnotesize {\bf (b)} H: 27.3 dB, 0.51 bpp}
\end{minipage} \hfill
\begin{minipage}{0.32\linewidth}
  \centering
  \includegraphics[width=\linewidth, trim=30mm 0mm 0mm 30mm, clip]{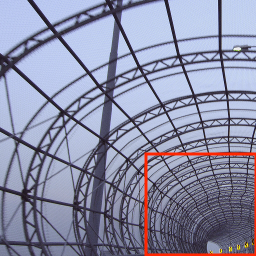} 
  \centerline{\footnotesize {\bf (c)} Original }
\end{minipage} \hfill
\caption{Sandwich results for high-resolution video transport using a lower-resolution codec (HEVC 4:4:4 LR.) (a) Sandwich, (b) HEVC. HEVC 4:4:4 LR is implemented as Bicubic downsampling, followed by HEVC 4:4:4 compression, followed by Lanczos3 upsampling. The interested reader can generate an extensive set of further examples using our software at \cite{sandwich_oss}.}
\vspace{-5mm}
\label{fig:video_hrlr_visual_results_rgb}
\end{figure}
\autoref{fig:video_hrlr_visual_results_rgb} compares the visual quality of the reconstructed sandwich clips to that of HEVC 4:4:4 LR. The INTER coded fifth frame of each clip is shown. In the first-row  note the significant amount of detail transported by the sandwich especially as depicted over the left side of the scene. HEVC 4:4:4 LR contains significant blur in those areas. The sandwich clip is 6 dB better at lower rate. In the second row note again not only the sharpness but the extra detail that the sandwich output contains especially toward the far-out points of the wire-structure. This detail, injected by the neural pre-processor and later demodulated by the neural post-processor, is simply missing from HEVC 4:4:4 LR. The sandwich is better by nearly 8 dB at the same rate.

\subsection{Compressing RGB Video with an Alternative Perceptual Metric (LPIPS)}
\label{sec:video-lpips}
\vspace{-1mm}
PSNR is well-known not to be a reliable metric for human-perceived visual quality. Alternatives such as  the Learned Perceptual Image Patch Similarity (LPIPS) 
\cite{zhang2018unreasonable, elpips, eero_lpips} 
and VMAF \cite{vmaf} are better suited for this task. While VMAF is not differentiable LPIPS is, and is thus readily suited to end-to-end gradient-based 
back-propagation. 
In this section we concentrate on LPIPS but also provide results on LPIPS-optimized sandwich networks on VMAF.

Since LPIPS is intended
for RGB, the final decoded YUV video is first converted
into RGB and then LPIPS is computed. In order to report results on
an approximately similar scale to mean-squared-error results, we derived a fixed linear scaler for
LPIPS so that for image vectors $x, y$,
\begin{eqnarray}
\vspace{-1mm}
\sLPIPS(x, y) \sim ||x-y||^2, \mbox{if } ||x-y||^2 < \tau
\vspace{-3.5mm}
\end{eqnarray}
where $\tau$ is a small threshold and $s$ is the LPIPS linear scaler. $s$ is calculated once and is fixed for all results.

As mentioned in \cite{elpips} when comparing single images there is the potential for adversarial attacks on LPIPS, \ie, an image $\tilde{x}$, unrelated to $x$ to a certain extent, can be designed to obtain $\LPIPS(\tilde{x}, y) \sim \LPIPS(x, y).$ One hence has to consider the possibility of an optimization scheme ``hacking'' the metric.  To that end \cite{elpips} recommends an ensemble LPIPS score where rather than calculating a single LPIPS score on the $x, y$ pair, one calculates several scores by randomly applying slight geometric and intensity transformations to the pair. These scores are then averaged for the ensemble score. This ensemble score is shown to be robust to adversarial attacks. 

\begin{figure}[t]
    \centering
\begin{minipage}{0.32\linewidth}
  \centering
  \includegraphics[width=\linewidth, trim=12mm 12mm 12mm 12mm, clip]{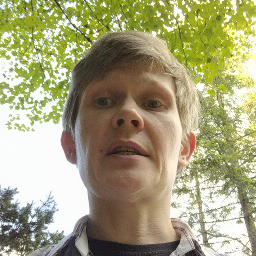}
  \centerline{\footnotesize {\bf (a)} S: 40.0 dB, 0.68 bpp}
\end{minipage} \hfill
\begin{minipage}{0.32\linewidth}
  \centering
  \includegraphics[width=\linewidth, trim=12mm 12mm 12mm 12mm, clip]{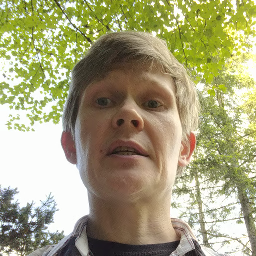}
  \centerline{\footnotesize {\bf (b)} H: 40.5 dB, 1.01 bpp}
\end{minipage} \hfill
\begin{minipage}{0.32\linewidth}
  \centering
  \includegraphics[width=\linewidth, trim=12mm 12mm 12mm 12mm, clip]{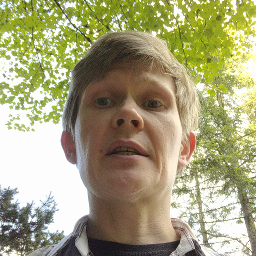}
  \centerline{\footnotesize {\bf (c)} Original}
\end{minipage} \hfill
    \caption{Sandwich results for the LPIPS scenario (5$^{th}$ frame in each clip is shown.). S: Sandwiched HEVC results, H: HEVC results. It is difficult to find significant differences among the results at 40 dB LPIPS-PSNR. The sandwich clip has $\sim$32\% lower rate than the HEVC clip.}
    \label{fig:hevc_444_lpips}
    \vspace{-4mm}
\end{figure}

By their very nature, video clips are typically versions of scenes with geometric deformations (scene motion) and color/brightness changes (scene lighting changes.)  
\begin{figure}[t]
    \centering
\begin{minipage}{0.32\linewidth}
  \centering
  \includegraphics[width=\linewidth, trim=12mm 12mm 12mm 12mm, clip]{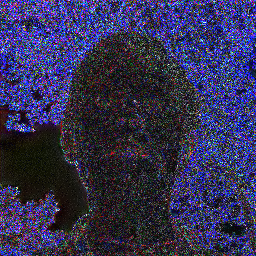}
  \centerline{\small {\bf (a)} S: 40.0 dB, 0.68 bpp}
\end{minipage}
\hspace{2em}
\begin{minipage}{0.32\linewidth}
  \centering
  \includegraphics[width=\linewidth, trim=12mm 12mm 12mm 12mm, clip]{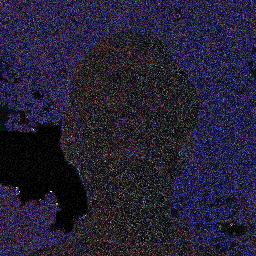}
  \centerline{\small {\bf (b)} H: 40.5 dB, 1.01 bpp}
\end{minipage} 
    \caption{Absolute errors (15x amplified) of the sandwich and HEVC frames depicted in \autoref{fig:hevc_444_lpips}. The sandwich frame has higher errors which are nevertheless difficult to perceive in \autoref{fig:hevc_444_lpips} as the visually important structures such as edges are well-preserved.}
    \vspace{-3mm}
    \label{fig:hevc_444_lpips_errors}
\end{figure}
We hence report the LPIPS loss of the $n^{th}$ clip as the average of the LPIPS losses over its $T=10$ frames, \ie,
\vspace{-1mm}
\begin{eqnarray}
D_n = \frac{1}{T} \sum_{t=0}^{T-1} \sLPIPS (x_n(t), y_n(t)),
\label{eq:lpips_clip}
\vspace{-3mm}
\end{eqnarray}
where $x_n(t), y_n(t)$ are the $t^{th}$ frames of the decoded and original clips respectively. We expect this averaging process to help improve the robustness of the score but we also (i) evaluate visual quality, (ii) show that the slim network obtains similar performance (with substantially reduced parameters the slim network has less room for hacking the metric,) and (iii) report VMAF results of the LPIPS-optimized sandwich on an especially meaningful scenario. In what follows we report ``LPIPS (RGB) PSNR'' which is the PSNR of the relevant averaged LPIPS loss.

\begin{figure}[t]
    \centering
\begin{minipage}{0.3\linewidth}
  \centering
  \includegraphics[width=\linewidth, trim=15mm 15mm 15mm 15mm, clip]{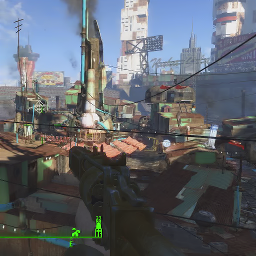}
  \centerline{\footnotesize {\bf (a)} S: 37.3 dB, 0.30 bpp}
\end{minipage} \hfill
\begin{minipage}{0.3\linewidth}
  \centering
  \includegraphics[width=\linewidth, trim=15mm 15mm 15mm 15mm, clip]{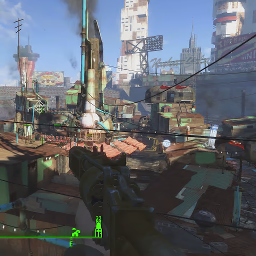}
  \centerline{\footnotesize  {\bf (b)} H: 37.4 dB, 0.48 bpp}
\end{minipage} \hfill
\begin{minipage}{0.3\linewidth}
  \centering
  \includegraphics[width=\linewidth, trim=15mm 15mm 15mm 15mm, clip]{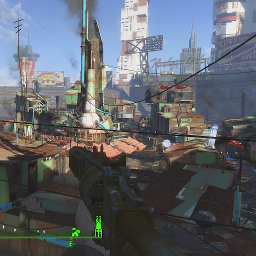}
  \centerline{\footnotesize  {\bf (c)} Original}
\end{minipage} \hfill
\\
\vspace{1mm}
\begin{minipage}{0.3\linewidth}
  \centering
  \includegraphics[width=\linewidth, trim=15mm 15mm 15mm 15mm, clip]{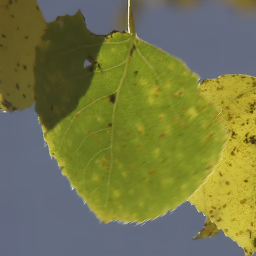}
  \centerline{\footnotesize {\bf (a)} S: 37.8 dB, 0.17 bpp}
\end{minipage} \hfill
\begin{minipage}{0.3\linewidth}
  \centering
  \includegraphics[width=\linewidth, trim=15mm 15mm 15mm 15mm, clip]{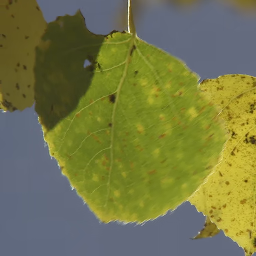}
  \centerline{\footnotesize {\bf (b)} H: 38.0 dB, 0.22 bpp}
\end{minipage} \hfill
\begin{minipage}{0.3\linewidth}
  \centering
  \includegraphics[width=\linewidth, trim=15mm 15mm 15mm 15mm, clip]{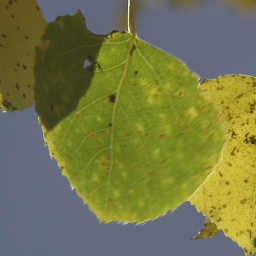}
  \centerline{\footnotesize {\bf (c)} Original}
\end{minipage} \hfill
\\
\vspace{1mm}
\begin{minipage}{0.3\linewidth}
  \centering
  \includegraphics[width=\linewidth, trim=15mm 15mm 15mm 15mm, clip]{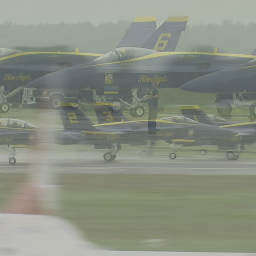}
  \centerline{\footnotesize {\bf (a)} S: 38.5 dB, 0.14 bpp}
\end{minipage} \hfill
\begin{minipage}{0.3\linewidth}
  \centering
  \includegraphics[width=\linewidth, trim=15mm 15mm 15mm 15mm, clip]{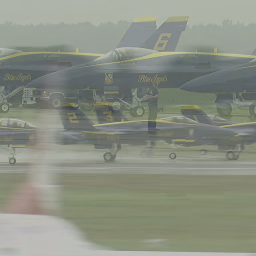}
  \centerline{\footnotesize {\bf (b)} H: 38.4 dB, 0.15 bpp}
\end{minipage} \hfill
\begin{minipage}{0.3\linewidth}
  \centering
  \includegraphics[width=\linewidth, trim=15mm 15mm 15mm 15mm, clip]{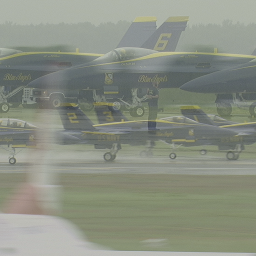}
  \centerline{\footnotesize {\bf (c)} Original}
\end{minipage}
    \caption{Qualifying the sandwich rate gains at same LPIPS quality. S: Sandwiched HEVC results, H: HEVC results. Top row: On video clips dense with high frequencies and textures the sandwich obtains the most significant gains ($\sim37.6\%$ for this clip).
    Middle row: On clips with lesser but still significant high-frequency content gains are reduced but remain significant ($\sim22.6\%$.)
    Bottom row: On clips showing smooth and blurry regions gains are further reduced ($\sim5.12\%$.) (See \autoref{fig:hevc_444_lpips_hi_mid_lo_large} for a larger version.)}
    \label{fig:hevc_444_lpips_hi_mid_lo_small}
    \vspace{-5mm}
\end{figure}

\autoref{fig:video_lpips_results_rgb} compares the rate-distortion performance of the neural-sandwiched HEVC against HEVC with clip distortion measured via (\ref{eq:lpips_clip}). Observe that the sandwich with the full model obtains $\sim$30\% improvements in rate at the same LPIPS quality. The slim model closely tracks these results with $\sim$20-25\% improvements. Lastly we see that a pre-processor only variant, where we have disabled the neural-post-processor 
and trained only a pre-processor, %
obtains $\sim$10-15\%. 
Considering that a new generation standard codec typically improves $\sim$30\% over the previous generation, one can see that {\em the full and slim networks are offering generational improvements assuming LPIPS accurately represents human-perceived quality}. As interestingly, the pre-processor-only result indicates that a video streaming service can potentially reduce bandwidth by 10-15\% in a  way transparent to its users' decoders. 

In order to vet the correspondence of LPIPS and visual quality we subjectively examined the clips of the neural-sandwiched HEVC and HEVC decoded at the same LPIPS quality. At high LPIPS quality levels we found no significant differences among sandwich, HEVC, and original clips. 
A sample from such a clip is shown in  \autoref{fig:hevc_444_lpips}, where all three samples look similar while the sandwich clip has $\sim$32\% less rate than the HEVC clip.  \autoref{fig:hevc_444_lpips_errors} in fact shows that the sandwich output contains more absolute-errors. It is nevertheless hard to  discern visual quality differences.
At intermediate to low quality levels it is difficult to pick between the sandwich and HEVC clips though differences to the original clip become more noticeable as seen in 
\autoref{fig:hevc_444_lpips_low_rates} in the supplementary \autoref{sec:supp-lpips}.
Quality degrades in expected ways and the sandwich retains rate improvements at the same LPIPS quality until the low quality regime. 
\autoref{fig:hevc_444_lpips_hi_mid_lo_small}
explores the types of clip and region statistics where LPIPS may be enabling the gains. See supplementary \autoref{sec:supp-lpips} for expanded discussion.

Last but not least, especially since VMAF is typically used to gauge quality in the streaming scenario, we also evaluated the VMAF scores of the sandwich pre-processor-only network. 
\autoref{fig:video_vmaf_results_rgb} in the supplementary shows that the pre-processor-only network (optimized for LPIPS with 10-15\% improvements) obtains $\sim$10\% improvements in rate at the same VMAF quality. While we have not done so this agreement between the metrics suggests training the sandwich for LPIPS, evaluating the models during training for VMAF as well, and picking a model that has acceptable improvements for both.

\vspace{-3mm}
\section{Discussion and Conclusion}
\label{sec:discussion}

In this paper, we have proposed sandwiching standard image and video codecs between neural pre- and post-processors,
jointly trained through a proxy.
Remarkably, the neural pre- and post-processor learn to communicate
{\em source} images to each other
by sending {\em coded} images through the standard codec.
The coded images may have fewer channels, lower resolution, and lower dynamic range than the source images that they represent.
Yet, even though the coded images are quantized by the standard codec to a PSNR commensurate with the bitrate, the source images achieve superior fidelity in the same distortion measure
or in an alternative distortion measure, at that bitrate.

While the sandwich architecture improves upon the standard codec's compression of typical color images under the MSE distortion, the strength of the %
architecture is that it allows the standard codec to adapt to coding non-typical images under
non-typical distortion measures.  
We have provided an extensive set of 
simulation results clearly demonstrating the value of the sandwich. Nevertheless our examples are not intended to be exhaustive.
We leave adaptation to medical, hyper-spectral, or other multi-channel imagery, to higher temporal resolution, and so on for future research.

Our experimental results consistently show that the pre- and post-processors jointly trained with our proxies can seamlessly be used to sandwich different codecs (JPEG, HEIC, HEVC) without retraining, obtaining
significant
rate-distortion gains compared to the non-sandwiched codec.
Furthermore the sandwich gains have been shown to extend to wrapping VVC and AV1 codecs with low-complexity models \cite{yueyu_icip}, with a level of complexity that can allow neural processing to be included in next generation video compression standards such as the upcoming AV2.
One can of course further improve these results by freezing the pre-processor after the generic-proxy-based training and fine-tune the post-processor using compressed data with a particular codec.

As standard codecs continue to advance 
one day the standard codec itself may be differentiable.  For example, it may become an end-to-end neural codec.  
It is important to note that even in that case, the sandwich architecture would remain a valid way to 
{\em re-purpose} the standard codec to alternative sources and distortion measures as pointed by \autoref{thm:proposition1}.  

We advocate that future image and video codec standards be designed with sandwiching in mind.  
As we have seen, it is possible to adapt codecs to altogether new image and distortion types. This can for example be accomplished by specifying a simple neural post-processor in the compressed bit stream header.
More generally, such a technique could be used at any level of the bit stream, \eg, GOP-level, picture-level, etc.,
to signal that specific neural processors be used in the decoder,  matched to specific neural processors used in the encoder, which have been trained to communicate %
through learned neural codes.  In these ways, we advocate for making standard codecs more universal and thus more broadly applicable to alternative source types and distortion measures, such as in computer graphics, augmented/virtual reality, medical imaging, multi-modal sensing for autonomous driving, and so forth. 

Our results clearly indicate that MSE-optimized codecs can be easily repurposed to other metrics and scenarios. 
Hence, despite the reputation of MSE as an inadequate visual quality metric, the call for replacing it with other metrics in standard codecs %
may not be clear-cut.
Given its ease of optimization in incrementally furthering individual compression tools it may in fact remain the metric of choice in designing inner compression engines to be generalized as needed by sandwiching.
\vspace{-8mm}

\bibliographystyle{IEEEtran}
\bibliography{sandwich_journal}

\clearpage
\pagebreak
\setcounter{page}{1}
\section{Supplementary Material}
\beginsupplement

\subsection{Theoretical Limits of Neural Sandwiching (Suppl.\ to Sec.~\ref{sec:prelude})}
\label{sec:theory}

In this section, we explore the theoretical limits of neural sandwiching.  In particular, we prove  Proposition~\ref{thm:proposition1_stronger}, which is stronger than Proposition~\ref{thm:proposition1} of \autoref{sec:prelude} in that it adds an arbitrary permutation.
We conclude the section with an algorithm for learning
the discussed optimal vector quantizer that accomplishes the optimal sandwiched codec from data.

\begin{proposition}[Stronger form of Proposition~\ref{thm:proposition1} - Optimal Sandwich]
\label{thm:proposition1_stronger}
Let $X$ be a $\mathbb{R}^n$-valued bounded source, let $d$ be a distortion measure, and let $D(R)$ be the operational distortion-rate function for $X$ under $d$.  For any  $\epsilon>0$, let $(\alpha^*,\beta^*,\gamma^*)$ be a rate-$R$ codec for $X$ achieving $D(R)$ within $\epsilon/2$.
Let
$(\alpha,\beta,\gamma)$ be a regular codec (\eg, a standard codec, possibly designed for a different source and different distortion measure) with bounded codelengths.  Then for any permutation $\pi$, there exist neural pre- and post-processors $f$ and $g$ such that the codec sandwich $(\alpha\circ f,g\circ\beta,\gamma)$ has expected distortion at most $D(R)+\epsilon$ and expected rate at most $R+D(p||q)+\epsilon$, where $p(k)=P(\{\alpha^*(X)=k\})$ and $q(k)=2^{-|\gamma(\pi(k))|}$.
\end{proposition}

\begin{remark}
$D(p||q)$ is the worst-case rate penalty for having to re-use the entropy coder from the standard codec.
\end{remark}

\begin{remark}
By optimizing over the permutation $\pi$, the rate penalty $D(p||q)$ may be minimized.  Indeed, the penalty is minimized when the codelengths $|\gamma(\pi(k))|$ are sorted in the same order as the codelengths $|\gamma^*(k)|$. (See Prop.~\ref{thm:sorting}.)
\end{remark}

First some definitions:  A $\mathbb{R}^n$-valued {\em source} $X$ is a random vector (\eg, an image), where $n$ is the dimension of the source (\eg, the number of pixels in the image).  The source is {\em bounded} if for some finite bound $b$, $X\in[-b,b]^n$ with probability 1. A {\em codec} for $X$ is given by a triple $(\alpha,\beta,\gamma)$, where the {\em encoder} $\alpha:\mathbb{R}^n\rightarrow\mathcal{K}$ maps each source vector $x\in\mathbb{R}^n$ to a index $k\in\mathcal{K}$, the {\em decoder} $\beta:\mathcal{K}\rightarrow\mathbb{R}^n$ maps each index $k\in\mathcal{K}$ to a reproduction vector $\hat x\in\mathbb{R}^n$; and the {\em lossless encoder} $\gamma:\mathcal{K}\rightarrow\mathcal{C}$ invertibly maps each element of $\mathcal{K}$ to a binary string in a {\em codebook} $\mathcal{C}\subset\{0,1\}^*$ of variable-length binary strings satisfying the Kraft inequality, $\sum_{k\in\mathcal{K}}2^{-|\gamma(k)|}\leq1$, where $|s|$ denotes the length in bits of the binary string $s$.  The Kraft inequality guarantees the existence of a prefix-free, and hence uniquely decodable (\ie, invertible), binary lossless encoder $\gamma$ \cite{cover2012elements}.  Alternatively $\gamma$ may be considered an arithmetic coder or other entropy coder with nominal codelengths $\{|\gamma(k)|\}$.  A {\em quantization cell} is the set $\{x:\alpha(x)=k\}$ of source vectors encoding to index $k$.  A codec is {\em regular} if each of its quantization cells is a non-degenerate polytope (\ie, the intersection of half-spaces with non-empty interior).  Codecs based on scalar quantization (\eg, transform coders) as well as nearest-neighbor quantizers are all regular.  A {\em distortion measure} $d:\mathbb{R}^n\times\mathbb{R}^n\rightarrow\mathbb{R}^+$ maps a source vector $x$ and its reproduction, say $\hat x=\beta(\alpha(x))$, to a non-negative number.  The expected {\em distortion} of the codec is
\begin{equation}
    D(\alpha,\beta) = E[d(X,\beta(\alpha(X)))],
\end{equation}
and the expected {\em rate} of the codec is
\begin{equation}
    R(\alpha,\gamma) = E[|\gamma(\alpha(X))|].
\end{equation}
The {\em operational distortion-rate function} for $X$ under $d$ is
\begin{equation}
    D(R) = \inf_{\alpha,\beta,\gamma} \left\{
    D(\alpha,\beta): R(\alpha,\gamma)\leq R]
    \right\}.
\end{equation}

\begin{proof}
Given the $\mathbb{R}^n$-valued bounded source $X$, the distortion measure $d$, the operational distortion-rate function $D(R)$ for $X$ under $d$, and any $\epsilon>0$, let $(\alpha^*,\beta^*,\gamma^*)$ be a near-optimal codec at rate $R$, such that
\begin{eqnarray}
    D(\alpha^*,\beta^*) & \leq & D(R) + \epsilon/2
    \label{eqn:D-alpha-star-beta-star} \\
    R(\alpha^*,\gamma^*) & \leq & R .
    \label{eqn:R-alpha-star-gamma-star}
\end{eqnarray}
Now given a regular codec $(\alpha,\beta,\gamma)$ generally not for the source $X$ but for some other source $Y$, which may be $\mathbb{R}^m$-valued, we need to find a neural pre-processor $f:\mathbb{R}^n\rightarrow\mathbb{R}^m$ and a neural post-processor $g:\mathbb{R}^m\rightarrow\mathbb{R}^n$ such that the composition $\alpha\circ f$ and the composition $g\circ\beta$ satisfy
\begin{eqnarray}
D(\alpha\circ f,g\circ\beta) & \leq & D(R) + \epsilon
\label{eqn:D-alpha-f-g-beta} \\
R(\alpha\circ f,\gamma) & \leq & R + D(p||q) + \epsilon
\label{eqn:R-alpha-f-gamma} ,
\end{eqnarray}
where $p(k)=P(\{\alpha(f(X))=k\})$ and $q(k)=2^{-|\gamma(\pi(k))|}$.

To accomplish this, we will find $f$ and $g$ so that $\alpha\circ f$ approximates the near-optimal encoder $\alpha^*$, and $g\circ\beta$ approximates the near-optimal decoder $\beta^*$.
Optionally we may also find a permutation $\pi:\mathcal{K}\rightarrow\mathcal{K}$ such that the composition $\gamma\circ\pi$ approximates the optimal lossless encoder $\gamma^*$.  Such a permutation minimizes the bound $D(p||q)$.

First we will prove (\ref{eqn:D-alpha-f-g-beta}) by showing that if the approximation is good enough, then the expected distortion increases by at most $\epsilon/2$, namely
\begin{equation}
    D(\alpha\circ f,g\circ\beta) \leq D(\alpha^*,\beta^*) + \epsilon/2 .
    \label{eqn:D-alpha-approx-beta-approx}
\end{equation}
Together, (\ref{eqn:D-alpha-star-beta-star}) and (\ref{eqn:D-alpha-approx-beta-approx}) imply (\ref{eqn:D-alpha-f-g-beta}).

To show (\ref{eqn:D-alpha-approx-beta-approx}), first we define an ``ideal'' pre-processor $f^*$ such that for all $x\in\mathbb{R}^m$, $f^*(x)=y_k$ whenever $\alpha^*(x)=k$, where $y_k$ is a point in the interior of the quantization cell $\{y:\alpha(y)=\pi(k)\}$.  (The cell has an interior because $(\alpha,\beta,\gamma)$ is assumed to be regular.)
We also define an ``ideal'' post-processor $g^*$ such that for all $k\in\mathcal{K}$: $g^*(\beta(\pi(k)))=\beta^*(k)$.  The definition of $g^*(\hat y)$ for
values of $\hat y$ not in the discrete set $\{\beta(\pi(k)): k \in \mathcal{K}\}$ are arbitrary, as $\beta$ produces values only in this set.
With these definitions, it can be seen that $\alpha(f^*(x))=\pi(\alpha^*(x))$ and
$g^*(\beta(\pi(\alpha^*(x))))=\beta^*(\alpha^*(x))$ and hence
\begin{equation}
    g^*(\beta(\alpha(f^*(x)))) = \beta^*(\alpha^*(x)) ,
\end{equation}
\ie, $\pi^{-1}\circ\alpha\circ f^*$ emulates $\alpha^*$ and $g^*\circ\beta\circ\pi$ emulates $\beta^*$.
Thus $D(\alpha\circ f^*,g^*\circ\beta)=D(\alpha^*,\beta^*)$.

We now argue that there exists a neural pre-processor $f$ sufficiently close to $f^*$. Indeed, by a Universal Approximation Theorem for neural networks \cite{enwiki:1190334119,NIPS2017_32cbf687,park2021minimum}, there is a neural network $f$ arbitrarily close to $f^*$ in $L_2$, i.e, for all $\delta>0$ there exists $f$ such that $E[||f(X)-f^*(X)||^2]<\delta$.  A fortiori, as convergence in $L_2$ implies convergence in probability, for all $\delta>0$, there exists $f$ and a set $\Omega$ with $P(\Omega)>1-\delta$ such that for all $x\in\Omega$, $||f(x)-f^*(x)|| < \delta$.  %
Thus whenever $x\in\Omega$ and $\alpha^*(x)=k$, by the definition of $f^*(x)$, we have $||f(x)-y_k||<\delta$. Since we have chosen $y_k$ to be in the interior of the cell $\{y:\alpha(y)=\pi(k)\}$, setting $\delta$ sufficiently small guarantees that $f(x)$ lies inside the cell $\{y:\alpha(y)=\pi(k)\}$ whenever $x\in\Omega$ and $\alpha^*(x)=k$.  That is, $\alpha(f(x))=\alpha(f^*(x))$ for all $x\in\Omega$.  In the unlikely event that $x\not\in\Omega$, there would be an encoding error.  But as the source is bounded, so is the distortion, by say $D_{\max}$.  Thus the expected distortion conditioned on an encoding error is at most $D_{\max}$.  Since $(1-P(\Omega))D_{\max}$ can be made less than $\epsilon/4$ by taking $\delta$ arbitrarily small, we have
\begin{eqnarray}
    \lefteqn{D(\alpha\circ f,g\circ\beta)} \\
    & = & E[d(X,g(\beta(\alpha(f(X)))))] \\
    & \leq & P(\Omega)E[d(X,g(\beta(\alpha(f(X)))))|\Omega] \\
    & & +(1-P(\Omega))D_{\max} \\
    & \leq & E[d(X,g(\beta(\alpha(f^*(X))))] + \epsilon/4 .
    \label{eqn:D-alpha-f-g-beta-part1}
\end{eqnarray}

We can use a similar argument to show the existence of a post-processor $g$ sufficiently close to $g^*$.  However, if the number of possible reproductions is finite (which is actually implied by our assumption that the codelengths are bounded), then a less sophisticated argument is needed, since then $g$ and $g^*$ need to be close only on a finite set of points.  In such case, it is clear that for
$\delta>0$, there exists $g$ such that for all $k\in\mathcal{K}$, $||g(\beta(\pi(k)))-g^*(\beta(\pi(k)))||<\delta$.  Hence by the continuity of the function $h(\hat x)=E[d(X,\hat x)]$ in $\hat x$,
for sufficiently small $\delta>0$ we have
\begin{eqnarray}
    \lefteqn{E[d(X,g(\beta(\alpha(f^*(X))))]} \\
    & = & E[d(X,g(\beta(\pi(\alpha^*(X))))] \\
    & \leq & E[d(X,g^*(\beta(\pi(\alpha^*(X))))] + \epsilon/4 \\
    & = & E[d(X,\beta^*(\alpha^*(X)))] + \epsilon/4 \\
    & = & D(\alpha^*,\beta^*) + \epsilon/4
    \label{eqn:D-alpha-f-g-beta-part2}
\end{eqnarray}
Together, (\ref{eqn:D-alpha-f-g-beta-part1}) and (\ref{eqn:D-alpha-f-g-beta-part2}) result in (\ref{eqn:D-alpha-approx-beta-approx}), and thus (\ref{eqn:D-alpha-f-g-beta}) is proved.

Next we prove (\ref{eqn:R-alpha-f-gamma}), by showing that if the approximation is good enough, then the expected rate increases by at most $D(p||q)$, namely
\begin{equation}
    R(\alpha\circ f,\gamma) \leq R(\alpha^*,\gamma^*) + D(p||q) + \epsilon .
    \label{eqn:R-alpha-approx-gamma}
\end{equation}
Together, (\ref{eqn:R-alpha-star-gamma-star}) and (\ref{eqn:R-alpha-approx-gamma}) imply (\ref{eqn:R-alpha-f-gamma}).

To show (\ref{eqn:R-alpha-approx-gamma}), we take a similar strategy to showing (\ref{eqn:D-alpha-approx-beta-approx}).  First, analogous to (\ref{eqn:D-alpha-f-g-beta-part1}), we have
\begin{eqnarray}
\lefteqn{R(\alpha\circ f,\gamma)} \\
& = & E[|\gamma(\alpha(f(X)))|] \\
& \leq & P(\Omega)E[|\gamma(\alpha(f(X)))||\Omega] \\
& & + (1-P(\Omega))R_{\max} \\
& \leq & E[|\gamma(\alpha(f^*(X)))|]+\epsilon .
\label{eqn:R-alpha-f-gamma-part1}
\end{eqnarray}
Then, analogous to (\ref{eqn:D-alpha-f-g-beta-part2}), we have
\begin{eqnarray}
\lefteqn{E[|\gamma(\alpha(f^*(X)))|]} \\
& = & E[|\gamma(\pi(\alpha^*(X)))|] \\
& = & \sum_k P(\{\alpha^*(X)=k\}) |\gamma(\pi(k))| \\
& = & -\sum_k p(k) \log_2 q(k) \\
& = & H(p) + D(p||q) \\
& \leq & R(\alpha^*,\gamma^*) + D(p||q) ,
\label{eqn:R-alpha-f-gamma-part2}
\end{eqnarray}
where $p(k)=P(\{\alpha^*(X)=k\})$ and $q(k)=2^{-|\gamma(\pi(k))|}$.
We have also used expressions for the entropy
\begin{equation}
    H(p)=-\sum_kp(k)\log_2p(k)
\end{equation}
and the Kullback-Leibler divergence
\begin{equation}
    D(p||q)=\sum_kp(k)\log_2\frac{p(k)}{q(k)} .
\end{equation}
Together, (\ref{eqn:R-alpha-f-gamma-part1}) and (\ref{eqn:R-alpha-f-gamma-part2}) result in (\ref{eqn:R-alpha-approx-gamma}), and thus (\ref{eqn:R-alpha-f-gamma}) is proved.

\end{proof}

Note that the permutation was not needed anywhere in the proof.  However we get it for free since the $y_k$s are arbitrary.  Moreover, we are now able to optimize over the permutation, to better approximate $\gamma^*$ with $\gamma\circ\pi$.  We now show:

\begin{proposition}
\label{thm:sorting}
For any given codec $(\alpha^*,\beta^*,\gamma\circ\pi)$, the minimum rate $R(\alpha^*,\gamma\circ\pi)$ is achieved when the codelengths $|\gamma(\pi(k))|$ have the same order as $-\log_2 P(\{\alpha^*(X)=k\})$.
\end{proposition}

\begin{proof}
An expression for the rate is
\begin{equation}
    R(\alpha^*,\gamma\circ\pi)
    = \sum_k P(\{\alpha^*(X)=k\}) |\gamma(\pi(k))| .
    \label{eqn:rate_expression}
\end{equation}
Thus if there exist $k_1$ and $k_2$ for which $P(\{\alpha^*(X)=k_1\})>P(\{\alpha^*(X)=k_2\})$ but $|\gamma(\pi(k_1))|>|\gamma(\pi(k_2))|$, then the rate (\ref{eqn:rate_expression}) can be strictly reduced by swapping $\gamma(\pi(k_1))$ and $\gamma(\pi(k_2))$, so that $|\gamma(\pi(k_1))|<|\gamma(\pi(k_2))|$.
\end{proof}

Note that since $\gamma^*$ minimizes the rate, the sequences $|\gamma^*(k)|$, $-\log_2 P(\{\alpha^*(X)=k\})$, and $|\gamma(\pi(k))|$ (the latter with a rate-minimizing permutation) all have the same order, up to ties.

Finally, 
we present an algorithm for learning an optimal Codelength Constrained Vector Quantizer (CCVQ) from data.  A CCVQ $(\alpha^*,\beta^*,\gamma^{(0)}\circ\pi)$ comprises an encoder $\alpha:\mathbb{R}^n\rightarrow\mathcal{K}$, a decoder $\beta:\mathcal{K}\rightarrow\mathbb{R}^n$, and a lossless codebook $\gamma:\mathcal{K}\rightarrow\{0,1\}^*$ minimizing the Lagrangian functional
\begin{equation}
    J_\lambda(\alpha,\beta,\gamma) = D((\alpha,\beta) + \lambda R(\alpha,\gamma)
\end{equation}
subject to $\gamma=\gamma^{(0)}\circ\pi$ being a reordering of a given invertible lossless codebook $\gamma^{(0)}$.  Our CCVQ design algorithm is similar to the Entropy Constrained Vector Quantization (ECVQ) design algorithm of \cite{ChouLG:89}, except that instead of assigning a codelength $-\log_2 P(\{\alpha(X)=k\})$ to index $k$, it must re-use one of the existing codelengths $|\gamma^{(0)}(k)|$.

First some notation.  For any set of values $v(k)$ indexed by $k\in\mathcal{K}$, let $\mathbf{k}=\mbox{argsort}(\{v(k):k\in\mathcal{K}\})$
denote a list of indices such that the $i$th element of the list, $\mathbf{k}_i$, is the index $k\in\mathcal{K}$ of the $i$th element of the set $\{v(k):k\in\mathcal{K}\}$ when sorted smallest to largest, with ties broken arbitrarily.
The algorithm is shown in Alg.~\ref{alg:algorithm}.

\begin{algorithm}[ht]
\caption{Codelength Constrained Vector Quantization}
\begin{algorithmic}[1]
\Require distribution $P$, distortion measure $d$, Lagrange multiplier $\lambda$, convergence threshold $\epsilon$, index set $\mathcal{K}$, decoder $\beta^{(0)}:\mathcal{K}\rightarrow\mathbb{R}^n$, lossless codebook $\gamma^{(0)}:\mathcal{K}\rightarrow\{0,1\}^*$
\State $t=0$, $J^{(0)}=\infty$, $\mathbf{k}^{(0)}=\mbox{argsort}(\{|\gamma^{(0)}(k)|:i\in\mathcal{K}\})$
\State $\forall x\!:\!\alpha^{(t+1)}(x)=\arg\!\min_{k\in\mathcal{K}}d(x,\beta^{(t)}(k))+\lambda|\gamma^{(t)}(k)|$
\State $\mathbf{k}^{(t+1)}=\mbox{argsort}(\{-\log P(\{\alpha^{(t+1)}(X)=k\}):k\in\mathcal{K}\})$
\State $\forall i:\gamma^{(t+1)}(\mathbf{k}_i^{(t+1)})=\gamma^{(0)}(\mathbf{k}_i^{(0)})$
\State $\forall k\!:\!\beta^{(t+1)}(k)=\arg\!\min_{\hat x}E[d(X,\hat x)|\alpha^{(t+1)}(X)\!=\!k]$
\State $D^{(t+1)}=E[d(X,\beta^{(t+1)}(\alpha^{(t+1)}(X)))]$
\State $R^{(t+1)}=E[|\gamma^{(t+1)}(\alpha^{(t+1)}(X))|]$
\State $J^{(t+1)}=D^{(t+1)}+\lambda R^{(t+1)}$
\If {$(J^{(t)}-J^{(t+1)})/J^{(t+1)}>\epsilon$} $t \gets t+1$ \& go to 2
\EndIf
\Ensure encoder $\alpha^{(t+1)}$, decoder $\beta^{(t+1)}$ and lossless codebook $\gamma^{(t+1)}$ minimizing Lagrangian functional $J_\lambda(\alpha,\beta,\gamma)$ subject to $\gamma$ being a reordering of $\gamma^{(0)}$
\end{algorithmic}
\label{alg:algorithm}
\end{algorithm}

\subsection{Quantizer and Rate Proxies (Suppl. to Sec.~\ref{sec:image_sandwich})}
\label{sec:supp-q-and-r-proxies}

Various differentiable {\em quantizer proxies} are possible (\autoref{fig:dist_proxy}).  Luo \etal\ \cite{luo2020ratedistortionaccuracy} use a soft quantizer $Q(X_i)$ whose transfer characteristic is a third-order polynomial spline.
Most end-to-end image compression works (\eg, \cite{ToOMHwViMi16}--\cite{MiBaTo18}) use either
additive uniform noise $Q(X_i)=X_i+U_i\Delta$, where $U_i$ is i.i.d. $\sim\text{unif}(-1/2,1/2)$,
or a ``straight-through'' quantizer $Q(X_i)=X_i+U_i\Delta$,
where $U_i\Delta =\text{stop\_gradient}(\text{round}(X_i/\Delta)-X_i/\Delta)\Delta$
is the true quantization noise and $\text{stop\_gradient}(\cdot)$
is the identity mapping but stops the gradient of its output from being back-propagated to its argument
\cite{stop_gradient}.
In all cases, the derivative of $\hat{X}_i=Q(X_i)$ with respect to $X_i$ is nonzero almost everywhere.  
This allows
non-trivial gradients of the end-to-end distortion $d(S,\hat{S})$ with respect to the parameters of the 
networks
using the chain rule and back-propagation.
These formulations also allow the 
stepsize $\Delta$ to receive gradients, which is necessary to properly minimize the Lagrangian.
We use straight-through quantization in our experiments.

\begin{figure}[b]
\vspace*{-2mm}
\footnotesize
\begin{minipage}[b]{0.22\linewidth}
  \centering
  \centerline{\includegraphics[width=1.5cm, trim=0 5.5in 850 0cm, clip]{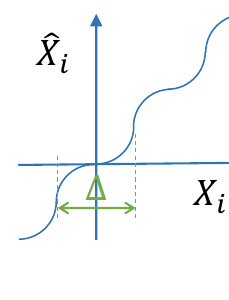}}
  \centerline{(a) Soft quantizer}\medskip
\end{minipage}
\begin{minipage}[b]{0.32\linewidth}
  \centering
  \centerline{\includegraphics[width=2.5cm, trim=0 5.5in 765 0cm, clip]{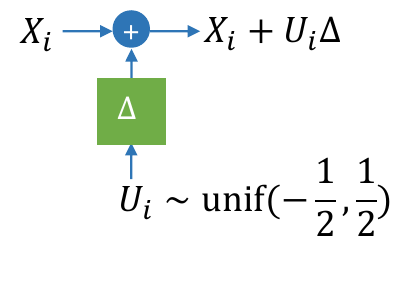}}
  \centerline{(b) Additive noise}\medskip
\end{minipage}
\begin{minipage}[b]{0.42\linewidth}
  \centering
  \centerline{\includegraphics[width=3.8cm, trim=0 5.125in 650 0cm, clip]{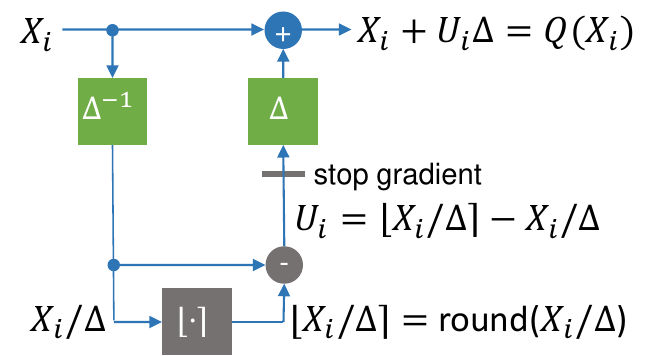}}
  \centerline{(c) Straight-through}\medskip
\end{minipage}
\vspace{-10pt}
\caption{Possible quantizer proxies.}
\label{fig:dist_proxy}
\end{figure}

Various differentiable {\em rate proxies} are also possible.  A convenient family of rate proxies $R(X)$ estimates the bitrate for a block of transform coefficients $X=[X_i]$ using affine functions of $\Norm{X}_2^2$, $\Norm{X}_1$, or $\Norm{X}_0$.
We focus on the latter in our experiments, since it is shown in \cite{He2001AUR} that an affine function of the number of nonzero quantized transform coefficients, $R(X)=a\sum_i \Indicator{ |x_i|\geq\Delta/2 }+b$, is an accurate rate proxy for transform codes.
In our work, we approximate the indicator function $\Indicator{|x_i|\geq\Delta/2}$ 
by the smooth differentiable function $\log(1+\Abs{x_i}/\Delta)$.
(An alternative would be to use $\tanh(\Abs{x_i}/\Delta)$.)
In sum, our rate proxy for a bottleneck image $B=[X^{(k)}]$ comprising multiple blocks $X^{(k)}$ is
\begin{equation}
    \!\!\!R(B) = \sum_k R(X^{(k)}) = a\sum_{k,i}\Log{ 1+\Abs{x_i^{(k)}}/\Delta } + b.
    \vspace{-.2cm}
\end{equation}
We set $b\!=\!0$ and determine $a$ for each bottleneck image~$B$
so that the rate proxy model matches the actual bitrate of the standard JPEG codec on that image, \ie:
\begin{equation}
    a = \frac{R_\text{JPEG}(B,\Delta)}{\sum_{k,i}\Log{1+\Abs{x_i^{(k)}}/\Delta}}.
\end{equation}
\vspace{-.2cm}

\noindent This ensures that the differentiable function $R(B)$ is exactly equal to $R_\text{JPEG}(B,\Delta)$ and that proper weighting is given to its derivatives on a per-image basis. Any image codec besides JPEG can also be used.
Similarly to the gradient of the distortion, the gradient of $R(B)$ with respect to the parameters of the pre-processor, and with respect to the stepsize $\Delta$, can be computed using back-propagation. An alternative rate proxy is given in \cite{SaidSP22}.

\subsection{High Resolution (HR) RGB Images with Lower Resolution (LR) Codecs (Suppl. to Sec.~\ref{sec:hr_lr_images})}
\label{sec:supp-hr_lr_images}
Bottleneck images corresponding to the results in \autoref{fig:SR_images}.
While the final reconstructions clearly depict the high-resolution detail
note  that the bottlenecks are patterned, aliased,  and  noisy. 
These are the manifestations of the neural codes that the processors use to communicate the high-resolution detail.
\begin{figure}[h]
  \includegraphics[width=0.32\linewidth, trim=10mm 0mm 10mm 20mm, clip]{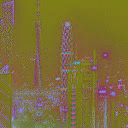}
  \includegraphics[width=0.32\linewidth, trim=10mm 0mm 10mm 20mm, clip]{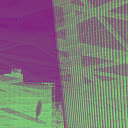}
  \includegraphics[width=0.32\linewidth, trim=10mm 0mm 10mm 20mm,  clip]{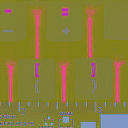}

  \caption{$128\!\times\!128$ reconstructed bottleneck images for the super-resolution sandwich results in \autoref{fig:SR_images} [enlarged for clarity]. Observe that while the bottlenecks appear aliased, noisy etc., the sandwich post-processor has correctly demodulated this ``noise'' in the final pictures.
  }
  \label{fig:SR_bottleneck_images1}
  \vspace{-5mm}
\end{figure}

\subsection{HR and HDR Adaptations (Suppl. to Sec.~\ref{sec:image_sandwich})}
\label{sec:supp-hr-and-hdr-adaptations}

In the HR problem, the RGB $H\times W\times 3$ source images have source bit depth $d=8$.  
Thus they have the standard dynamic range, $[0,255]$.  However, the 
bottleneck images have lower spatial resolution, $H/2 \times W/2 \times 3$. 
In our work, the resampler in the pre-processor comprises bicubic filtering and 2x downsampling; the resampler in the post-processor comprises Lanczos3 interpolation of the half-resolution images back to full-resolution. Similar down- and up-sampling is done in \cite{EusebioAP20}.

In the HDR problem, the source images 
have dynamic range $\left[0,2^d-1\right]$, where $d$ is the source bit depth. 
The bottleneck images have dimensions that match the source images: $H\times W\times 3$
but
are restricted to the standard dynamic range $[0, 255]$.  Since the codec proxy does not pass any information outside of this range, the pre-processor produces images in this range.

In both the HR and HDR problems, the sandwiched codecs operate in 4:4:4 mode without a color transform.  %
Regardless,
the baseline (non-sandwiched) codecs that we compare to use the RGB $\leftrightarrow$ YUV transform when it is beneficial for them in an R-D sense: In the HR scenario they use the color transform; in HDR they encode RGB directly.

\subsection{Further Normal Map Image Results (Suppl.\ to Sec.~\ref{sec:normal_maps})}

\autoref{fig:results_normals2} shows results of compressing normal maps with sandwiched JPEG.  Compare \autoref{fig:results_normals1}, which shows corresponding results for sandwiched HEIC.  The gains due to sandwiching are preserved in either case.

\begin{figure}[h]
\centering
  \includegraphics[width=0.8\linewidth, trim=11mm 4mm 15mm 8mm, clip]{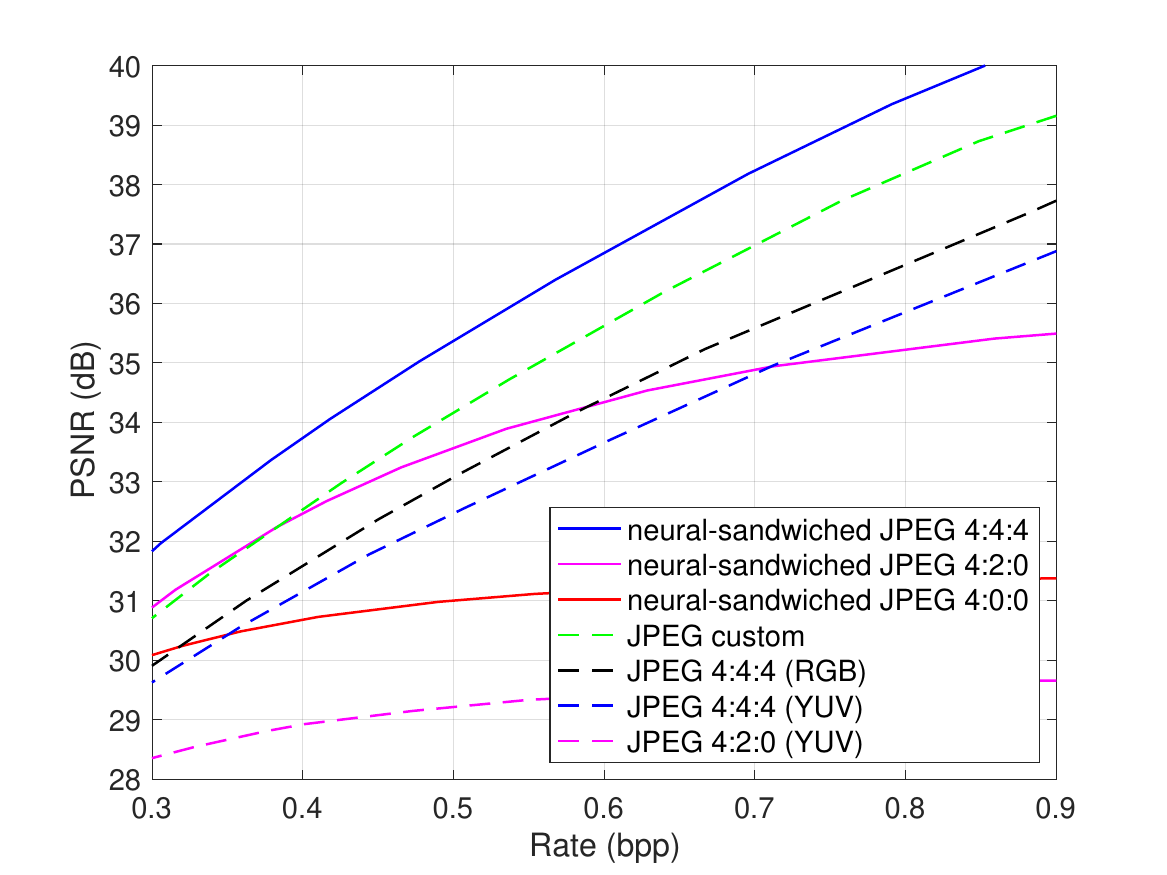}
  \caption{R-D performances of compressing normal map images with JPEG and neural-sandwiched JPEG, in various formats.}
  \label{fig:results_normals2}
\end{figure}

\subsection{Further LPIPS Video Results (Suppl.\ to Sec.~\ref{sec:video-lpips})}
\label{sec:supp-lpips}

\autoref{fig:video_vmaf_results_rgb} shows that the pre-processor-only network (optimized for LPIPS with 10-15\% improvements, see \autoref{fig:video_lpips_results_rgb}) obtains $\sim$10\% improvements in rate at the same VMAF quality. While we have not done so this agreement between the metrics suggests training the sandwich for LPIPS, evaluating the models during training for VMAF as well, and picking a model that has acceptable improvements for both.
\begin{figure}[h]
\centering
  \includegraphics[width=9.0cm]{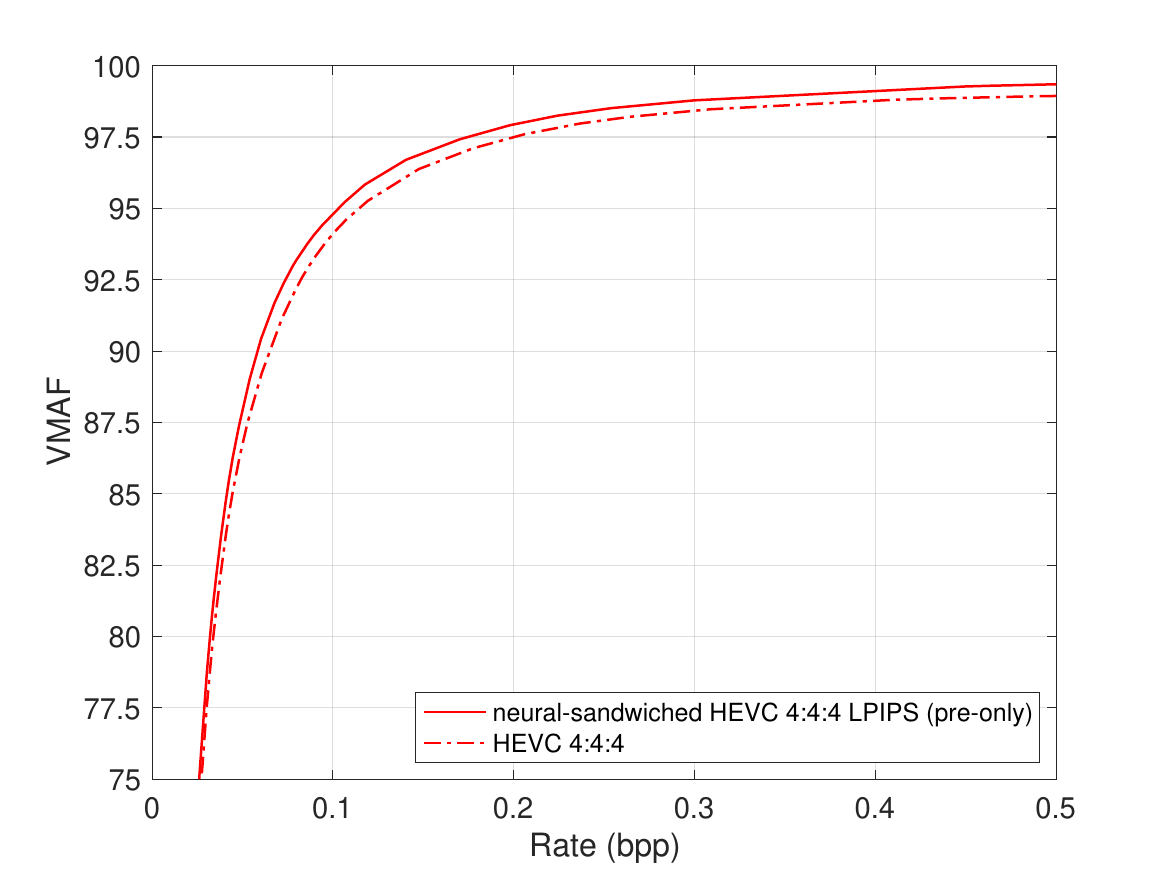}
\vspace{-3mm}
  \caption{VMAF scores of the sandwich trained using LPIPS (pre-processor only) and HEVC. Over a broad range sandwich shows $\sim10\%$ improvements in rate at same VMAF quality.}
  \label{fig:video_vmaf_results_rgb}
  \vspace{-3mm}
\end{figure}

\autoref{fig:hevc_444_lpips_low_rates} provides further examples of using a sandwich optimized for LPIPS to transport images through a standard codec.  Whereas \autoref{fig:hevc_444_lpips} showed 32\% rate savings at a quality visually close to the original, \autoref{fig:hevc_444_lpips_low_rates} shows performance at lower visual quality levels. 
As the quality is lowered it is difficult to pick between the sandwich and HEVC clips though differences to the original clip become more noticeable. Through many such visualizations we observed that the LPIPS-optimized sandwich matched or exceeded HEVC visual quality with often times very significant savings in rate.

To better understand where the sandwich with LPIPS is getting improvements we considered clips where rate gains (at the same LPIPS quality) were high, intermediate, and low. \autoref{fig:hevc_444_lpips_hi_mid_lo_large} shows sample clips from each group. As illustrated the gains are strongly tied to the high frequency and texture content of the scene. While scenes dense with such content (top row) have the largest improvements, scenes with even moderate amounts of high frequency structures (middle row) induce noticeable gains.

\begin{figure*}[p]
    \centering
\begin{minipage}{0.48\linewidth}
  \centering
  \includegraphics[width=\linewidth]{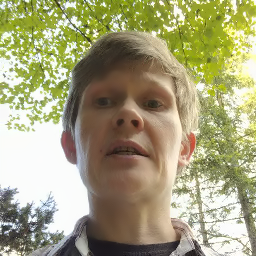}
  \centerline{\small {\bf (a)} Sandwich (34.5 dB, 0.34 bpp)}
\end{minipage}
\begin{minipage}{0.48\linewidth}
  \centering
  \includegraphics[width=\linewidth]{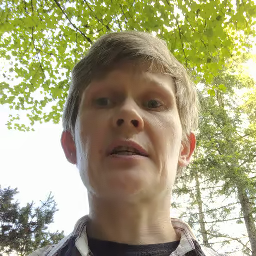}
  \centerline{\small {\bf (b)} HEVC (34.4 dB, 0.44 bpp)}
\end{minipage} 
\\
\begin{minipage}{0.48\linewidth}
  \centering
  \includegraphics[width=\linewidth]{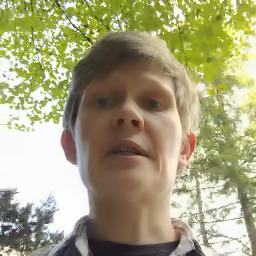}
  \centerline{\small {\bf (c)} Sandwich (30.4 dB, 0.11 bpp)}
\end{minipage}
\begin{minipage}{0.48\linewidth}
  \centering
  \includegraphics[width=\linewidth]{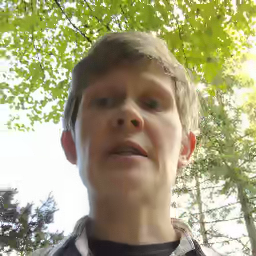}
  \centerline{\small {\bf (d)} HEVC (30.3 dB, 0.11 bpp)}
\end{minipage} 
    \caption{Comparison of sandwich and HEVC clips at lower rates for the scenario in \autoref{fig:hevc_444_lpips}. As the rate is lowered both the sandwich and HEVC clips have lower LPIPS-PSNR and visual quality. At the same LPIPS quality, it is still difficult to have a firm preference between them. In the top row rate is approximately half that of
    \autoref{fig:hevc_444_lpips}. Sandwich is still better by $\sim20\%$ in rate. The bottom row shows the low quality regime. Both clips lose texture detail but HEVC seems to have more artifacts. }
    \label{fig:hevc_444_lpips_low_rates}
\end{figure*}

\begin{figure*}[p]
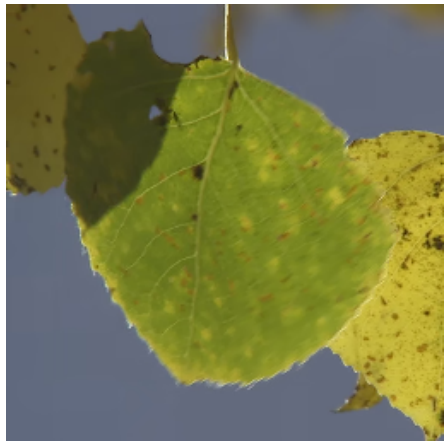
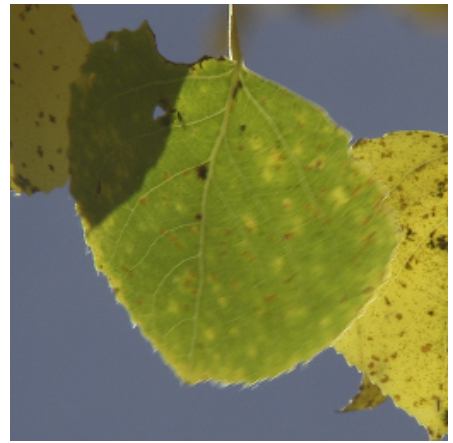
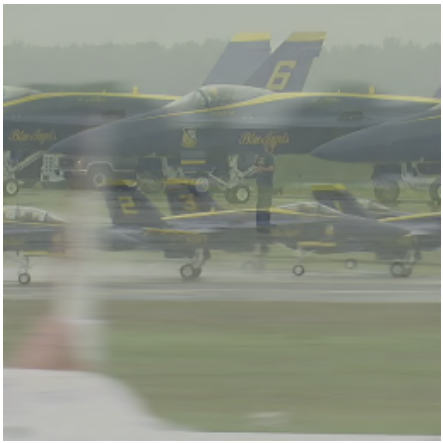
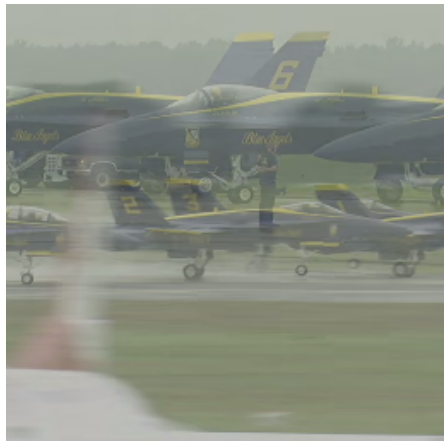
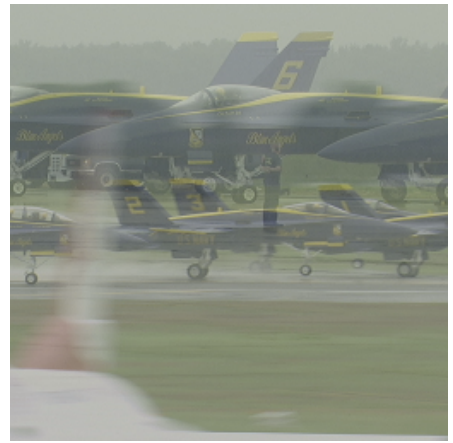

    \centering
\begin{minipage}{0.32\linewidth}
  \centering
  \includegraphics[width=\linewidth]{yuv_dataset/LPIPS/plot_0022_sw_lpips_D=37.3dB_R=0.30bpp_37.61pc_04.png}
  \centerline{\small {\bf (a)} Sandwich (37.3 dB, 0.30 bpp)}
\end{minipage} \hfill
\begin{minipage}{0.32\linewidth}
  \centering
  \includegraphics[width=\linewidth]{yuv_dataset/LPIPS/plot_0022_hevc_lpips_D=37.4dB_R=0.48bpp_04.png}
  \centerline{\small {\bf (b)} HEVC (37.4 dB, 0.48 bpp)}
\end{minipage} \hfill
\begin{minipage}{0.32\linewidth}
  \centering
  \includegraphics[width=\linewidth]{yuv_dataset/LPIPS/plot_0022_original_04.png}
  \centerline{\small {\bf (c)} Original)}
\end{minipage} \hfill
\\
\vspace{1mm}
\begin{minipage}{0.32\linewidth}
  \centering
  \includegraphics[width=\linewidth]{yuv_dataset/LPIPS/plot_0041_sw_lpips_D=37.8dB_R=0.17bpp_22.68pc_04.png}
  \centerline{\small {\bf (a)} Sandwich (37.8 dB, 0.17 bpp)}
\end{minipage} \hfill
\begin{minipage}{0.32\linewidth}
  \centering
  \includegraphics[width=\linewidth]{yuv_dataset/LPIPS/plot_0041_hevc_lpips_D=38.0dB_R=0.22bpp_04.png}
  \centerline{\small {\bf (b)} HEVC (38.0 dB, 0.22 bpp)}
\end{minipage} \hfill
\begin{minipage}{0.32\linewidth}
  \centering
  \includegraphics[width=\linewidth]{yuv_dataset/LPIPS/plot_0041_original_04.png}
  \centerline{\small {\bf (c)} Original)}
\end{minipage} \hfill
\\
\vspace{1mm}
\begin{minipage}{0.32\linewidth}
  \centering
  \includegraphics[width=\linewidth]{yuv_dataset/LPIPS/plot_0044_sw_lpips_D=38.5dB_R=0.14bpp_5.12pc_04.png}
  \centerline{\small {\bf (a)} Sandwich (38.5 dB, 0.14 bpp)}
\end{minipage} \hfill
\begin{minipage}{0.32\linewidth}
  \centering
  \includegraphics[width=\linewidth]{yuv_dataset/LPIPS/plot_0044_hevc_lpips_D=38.4dB_R=0.15bpp_04.png}
  \centerline{\small {\bf (b)} HEVC (38.4 dB, 0.15 bpp)}
\end{minipage} \hfill
\begin{minipage}{0.32\linewidth}
  \centering
  \includegraphics[width=\linewidth]{yuv_dataset/LPIPS/plot_0044_original_04.png}
  \centerline{\small {\bf (c)} Original)}
\end{minipage}
    \caption{Qualifying the sandwich rate gains at same LPIPS quality. Top row: On video clips dense with high frequencies and textures the sandwich obtains the most significant gains ($\sim37.6\%$ for this clip).
    Middle row: On clips with lesser but still significant high-frequency content gains are reduced but remain significant ($\sim22.6\%$.)
    Bottom row: On clips showing smooth and blurry regions gains are further reduced ($\sim5.12\%$.) }
    \label{fig:hevc_444_lpips_hi_mid_lo_large}
\end{figure*}

\end{document}